\documentclass{article}

\usepackage[english]{babel}

\usepackage[letterpaper,top=2cm,bottom=2cm,left=3cm,right=3cm,marginparwidth=1.75cm]{geometry}

\usepackage{graphicx}
\usepackage{xcolor}
\usepackage{float}
\usepackage{array}  
\usepackage{amssymb}
\usepackage{amsthm}
\usepackage{amsfonts,amssymb}
\usepackage{amsmath}
\theoremstyle{definition}
\newtheorem{definition}{Definition}
\newtheorem{thm}{Theorem} 

\newtheorem{cor}[thm]{Corollary}
\newtheorem{prop}{Proposition}
\usepackage[colorlinks=true, allcolors=blue]{hyperref}
\def\a{\alpha} \def\b{\beta} \def\g{\gamma} \def\d{\delta}

\def\bY{\pmb {\rm Y}}
\def\bI{\pmb {\rm I}}
\def\bG{\pmb \Gamma}
\def\bZ{\pmb {\rm Z}}

\def\third{\text{\footnotesize $\frac 13$}}

\numberwithin{equation}{section}

\title{Acoustic axes conditions revised}
\author{Yakov Itin\\
Jerusalem College of Technology, Jerusalem, Israel.}   

\begin{document}
\maketitle

\begin{abstract}
The explanation of the basic acoustic properties of crystals requires a recognition of the acoustic axes. To derive the acoustic axes in a given material, one requires both a workable method and the necessary and sufficient criteria for the existence of the acoustic axes in a partial propagation direction. 
We apply the reduced form of the acoustic tensor to the acoustic axis conditions in the present work. Using this tensor, we obtain in a compact form, allowing for qualitative analysis, the necessary and sufficient criteria for the existence of the acoustic axis.   Furthermore, the well-known Khatkevich criteria and their variants are recast in terms of the reduced acoustic tensor. This paper's primary input is an alternate minimal polynomial-based system of acoustic axes conditions. In this approach, we derive an additional characteristic of acoustic axes: the directions in which the minimal polynomial of the third order is reduced to that of the second order. Next, we offer a general solution to the second-order minimum polynomial equation, that utilizes a scalar and a unit vector for defining the acoustic tensor along the acoustic axis. It is shown that the scalar matches the eigenvalue of the reduced acoustic tensor, and the vector corresponds to the polarization into the single eigenvalue direction. We use the minimal polynomial construction to demonstrate the equivalence of different acoustic axis criteria. We demonstrate the applicability of this approach to actual computations of the acoustic axes and their fundamental properties (phase speeds and polarizations) for high symmetry cases, such as isotropic materials and RTHC crystals.
\end{abstract}

\section{Introduction}
Acoustic wave propagation in the elastic medium is represented by three distinct waves, each with its own phase speed. In anisotropic media, these three speeds are typically different and depend on the spatial directions. In the majority of materials, nevertheless, there are specific directions where at least two waves have identical speeds.  The terms ``{\it acoustic axes}" refer to these specific directions. The acoustic axes are distinguished in several aspects.  The strict hyperbolicity of the system of the associated partial differential equations is lost along them, in particular. 
Furthermore, these directions show anomalies in the polarization field and degeneracy of the elasticity waves, see Alshits  \&  Lothe \cite{Alshits2} and Domanski \cite{Domanski}. Moreover, the acoustic axes have been shown to be associated with energy focussing; see Every \cite{Every1986}, Hauser, Weaver, \& Wolfe \cite{Hauser}, Vavry\v{c}uk \cite{Vavrycuk}, and the reference given therein. The latter aspect makes acoustic axes relevant in a wide range of acoustic applications.

The theoretical study of the acoustic axes may be addressed in numerous issues: 
\begin{itemize}
    \item Formulation of the necessary and sufficient conditions for the existence of acoustic axis in a given direction; see Fedorov  \cite{Fedorov}, Norris \cite{Norris}, Alshits \& Lothe \cite{Alshits}, \cite{Alshits1}, also  \cite {It-axes1}.  These conditions must be invariant under arbitrary rotations of the coordinate basis. 
    \item Study of geometric and topological features of the the acoustic axes and the associated polarization fields. These problems include the study of the existence, number, and types of acoustic axes in the general case as well as in the specific types of materials (mainly crystals); see Alshits \& Lothe \cite {Alshits}, \cite{Alshits1}, \cite{Alshits2}, Domanski \cite{Domanski}, Holm \cite{Holm}, Vavry\v{c}uk \cite{Vavrycuk}, and the references cited therein. 
    \item The development of a practical algorithm for deriving the acoustic axis directions. Such conditions do not necessary need to be invariant or independent. They must, however, be simple to implement. 
    Starting from the pioneering pair of conditions proposed by Khatkevich \cite{Khatkevich}, different systems of conditions of such type were derived by Mozhaev et al. \cite{Mozhaev}, Norris \cite{Norris}, Alshits \& Lothe \cite{Alshits}, \cite{Alshits1}.  
    \item Actual computation of the acoustic axes in natural crystals; see e.g. Fedorov \& Fedorov \cite{Fed-Fed}, Boulanger \& Hayes \cite{Boulanger}, \cite{Boulanger1}, and the references cited therein. 
\end{itemize}

Our objective in this paper is to study the theoretical aspects of the acoustic axes problem using the reduced form of the acoustic tensor, see \cite{It-axes1}. We specifically provide a compact form expression for the necessary and sufficient condition for the presence of the acoustic axis in a given spatial direction.  Because these criteria are stated by an algebraic equation of the 12th order, they are useless for the majority of quantitative applications. We reformulate the known 6th-order Khatkevich conditions as well as their various version in terms of the reduced acoustic tensor. An alternative system of 4th-order conditions based on minimal polynomials is our main finding. In this method, we derive an additional feature of acoustic axes: They are shown to be the unique directions in which the minimal polynomial of the third order is reduced to that of the second order. We prove the equivalence of several kinds of conditions, including those expressed as systems of polynomial equations with different orders. 
The key element of this proof is the minimal  polynomial condition. 
Next, we offer a general solution to the second-order minimal  polynomial equation, which uses a scalar and a unit vector to represent the acoustic tensor along the acoustic axis. It is proved that the vector corresponds to the polarization into the single eigenvalue direction, and the scalar coincides with this eigenvalue. Such settings may be advantageous for actual computations of the acoustic axes and their basic properties (phase speeds and polarizations). 
We illustrate this method with simple examples of isotropic media and cubic crystals as well as with a more involved example of the RTHC crystals. 

The present paper is organized as follows: We give a brief description of the reduced Christoffel (acoustic) tensor and outline the fundamental formulas in Section 2. The adjoint matrix approach, which forms the foundation of the Khatkevich method and its variations, is covered in Section 3. The correspondence constructs are reformulated using the reduced Christoffel tensor. We provide the new minimal polynomial construction and discuss how it relates to the adjoint matrix technique in Section 4. The minimal polynomial equation solution and polarization vector representation of the reduced acoustic tensor are covered in Section 5. The method is illustrated in Section 6 through the computation of the acoustic axes for cubic crystals and isotropic materials.  Exact formulas for the directions of the acoustic axes are obtained also by studying a more complicated version of the RTHC crystals.  The Conclusion section includes a brief discussion of our key results and their potential extensions.

Notations: We indicate the components of the tensors in the regular font and use the bold font for the tensors and vectors presented in the coordinate-free notations. In the tensorial formulas, the Roman indices vary in the range  $i,j,\cdots=(1,2,3)$. We employ the upper and lower indices summing following Einstain's summation rule in the covariant formulations. 

\section{Acoustic axis---discriminant  condition}
Acoustic waves in an anisotropic medium are determined by the Christoffel (acoustic) tensor.
The components of this tensor are expressed in terms of the elasticity tensor $C^{ijkl}$, the propagation vector ${\bf n}=(n_1,n_2,n_3)$, and the mass density $\rho$,
\begin{equation}\label{ac-ax2}
    \Gamma^{il}:=\frac 1\rho\, C^{ijkl}n_jn_k\,.
\end{equation}
The Christoffel tensor is symmetric under the exchange of its indices, 
\begin{equation}\label{ac-ax3}
\Gamma^{ij}=\Gamma^{ji},
\end{equation}
because of the symmetry of the elasticity tensor. 
Then, in general, it can be represented as a symmetric matrix $\bG$ consisting of six independent components that are quadratic polynomials of the propagation vector, $\bG={\bG(\bf n)}$. 
 In terms of the acoustic tensor,  propagation of the acoustic wave in a homogeneous medium is 
described by an ordinary eigenvalue system of three linear equations: 
\begin{equation}\label{ac-ax4}
\left( v^2\, {\pmb {\rm I}}-{\pmb\Gamma}\right){\bf U}=0\,,\qquad {\rm or}\qquad \left( v^2g^{il}-\Gamma^{il}\right)U_l=0\,.
\end{equation}
Here, the scalar $v$ denotes the phase speed, whereas the vector ${\bf U}$ denotes the wave polarization. 

For a given direction ${\bf n}$, system (\ref{ac-ax4}) has a nontrivial polarization vector ${\bf U}$ as a solution if and only if the characteristic (secular) equation 
\begin{equation}\label{ac-ax5}
\det\left(v^2\, {\pmb {\rm I}}-{\bG}\right)=0\,
\end{equation}
holds. Since $\bG$ is represented in a chosen coordinate system by a symmetric third-order matrix, it has three real eigenvalues $\lambda=v^2$. Correspondingly, three independent acoustic waves exist in the direction ${\bf n}$ providing the matrix ${\pmb\Gamma}$ is positive definite in this direction. 
For the third-order tensor $\bG$, the Eq.(\ref{ac-ax5}) can be written in terms of the variable $\lambda=v^2$ as a general cubic equation
\begin{equation}\label{ac-ax6}
f(\lambda)=\lambda^3+M\lambda^2+P\lambda+Q=0.
\end{equation}
The coefficients of this cubic are independent of the coordinate system in use. Moreover, they are expressed by the fundamental invariants of the matrix $\bG$,
\begin{equation}\label{ac-ax7}
M=-{\rm tr }\,{\bG},\qquad P=\frac 12\left({\rm tr }^2\bG-{\rm tr }\,\bG^2\right),\qquad Q=-{\rm det }\,\bG. 
\end{equation}
Since the components of the matrix $\bG={\bG(\bf n)}$ are the quadratic homogeneous polynomials in the vector ${\bf n}$, we are dealing in (\ref{ac-ax7}) with homogeneous polynomials in  ${\bf n}$ of order $2,4,$ and 6 for the coefficients $M,P,Q$, respectively. In the case of three distinct solutions of the characteristic equation (\ref{ac-ax6}) the triad of orthonormal polarization vectors $\{{\bf U}^{(1)},{\bf U}^{(2)},{\bf U}^{(3)}\}$ is uniquely definite. 

{\it Acoustic axes} are defined as the propagation directions ${\bf n}$ where two or three acoustic waves have an equal phase speed; see Khatkevich 
 \cite{Khatkevich},  Fedorov \cite{Fedorov}, etc.  The polarization vector ${\bf U}^{(3)}$ is uniquely definite in the situation of only two equal speeds $v^{(1)}=v^{(2)}\ne v^{(3)}$, whereas the other two vectors, ${\bf U}^{(1)}$  and ${\bf U}^{(2)}$, can be selected arbitrarily from the plane perpendicular to ${\bf U}^{(3)}$. In the case of three equal speeds the polarization vectors $\{{\bf U}^{(1)},{\bf U}^{(2)},{\bf U}^{(3)}\}$ can be chosen arbitrarily from the unit sphere.  
 In both cases, one speaks on the {\it degenerate wave propagation} along the directions ${\bf n}$; see \cite{Alshits1}. 
 
 Apart from the acoustic axes, it is useful to consider the additional  ``special directions" defined by Fedorov \cite{Fedorov}. 
In particular, ${\bf n}$ is referred to as a {\it pure longitudinal direction} if one of the polarization vectors happens to be parallel to it, i.e., ${\bf n}={\bf U}$.  The term {\it pure shear direction} refers to ${\bf n}$ when there is a polarization vector that is perpendicular to it, ${\bf n}\bot{\bf U}$. Since three polarizations are mutually orthogonal, all longitudinal direction  are shear directions; the opposite is not true. 

Considering the acoustic axes, we conclude that they are all shear directions. Indeed, let ${\bf n}$ be an acoustic axis with ${\bf U}^{(3)}$ corresponding to the single eigenvalue and ${\bf U}^{(1)},{\bf U}^{(2)}$ corresponding to the double eigenvalue. Denote the plane spanned by the vectors ${\bf U}^{(1)}$ and ${\bf U}^{(2)}$ by $L_1$. This plane is degenerate---every vector lying in $L_1$ can serve as a polarization corresponding to the same double eigenvalue. Denote the plane perpendicular to ${\bf n}$ by $L_2$. It is necessary for the two subspaces $L_1$ and $L_2$ to intersect (or coincide) in the three-dimensional vector space. In addition to being an acoustic axis, a vector lying in the intersection $L_1\cap L_2$ is also perpendicular to $\bf n$. Then $\bf n$ is a  shear direction.  The opposite  is not true, i.e, in general, there are shear directions that are not acoustic axes.

A key issue in the current topic is to establish whether an acoustic axe exists in the propagation direction ${\bf n}$. 
From the theory of cubic equations,  the { necessary and sufficient condition} for the existence of two equal roots for Eq.(\ref{ac-ax6}) is well known. It is determined by the cubic discriminant vanishing, which takes the following form: 
\begin{equation}\label{ac-ax8}
\Delta=M^2P^2-4P^3-4M^2Q^2+18MPQ-27Q^2=0.
\end{equation}
This polynomial equation, dealing with the two independent components of the unit vector ${\bf n}$, is of the 12th order. 
In \cite{Fedorov}, Fedorov suggested Eq.(\ref{ac-ax8}) as a {\it necessary and sufficient condition} for the existence of the acoustic axe in the direction ${\bf n}$. For the qualitative analysis and the actual derivation of the acoustic axis in materials, this condition is unworkable, even though it completely solves the problem from a purely mathematical perspective. 

A simple way for streamlining the necessary and sufficient condition (\ref{ac-ax8}) was proposed in \cite{It-axes1}.
Assemble the scalar and traceless components of the acoustic tensor, 
\begin{equation}\label{ac-ax9}
{\pmb \Gamma}={\pmb {\rm Y}}+\gamma \,{\pmb {\rm I}}, \qquad {\rm or}\qquad \Gamma^{ij}=Y^{ij} +\gamma g^{ij},
    \end{equation}
    where 
    \begin{equation}
 \gamma:=\frac 13 \,{\rm tr}\,\,{\pmb \Gamma}=\frac 13 g_{ij}\Gamma^{ij}=\frac 13\left(\Gamma^{11}+\Gamma^{22}+\Gamma^{33}\right)\,\,.
    \end{equation}
Then we introduced the {\it reduced Christoffel tensor} ${\pmb {\rm Y}}:={\pmb \Gamma}-\gamma \,{\pmb {\rm I}}$ which not only meets the fundamental characteristics of the acoustic tensor ${\pmb {\rm \Gamma}}$ but also has zero trace. 
Compared to the full Christoffel tensor $\bG$, which consists of six independent components, the reduced Christoffel tensor $\bY$ consists of only five independent components. 

Substituting (\ref{ac-ax9}) into Eq.(\ref{ac-ax4}), we obtain 
\begin{equation}\label{ac-ax10}
\left(\lambda  g^{il}-Y^{il} -\gamma g^{il}\right)U_l=0\,.
\end{equation}
We now recast Eq. (\ref{ac-ax10}) as an ordinary eigenvalue problem for the reduced Christoffel tensor, introducing the {\it shifted eigenvalues} 
\begin{equation}
    \rho:=\lambda-\gamma.
\end{equation}
Then we Eq.(\ref{ac-ax10}) tales the standard form
\begin{equation}\label{ac-ax11}
\left(\rho\, {\pmb {\rm I}}-{\pmb {\rm Y}}\right){\pmb {\rm U}}=0,\qquad{\rm or}\qquad\left(\rho g^{il}-Y^{il}\right)U_l=0 \,.
\end{equation}
The characteristic equation (\ref{ac-ax5}) reads now as
\begin{equation}\label{ac-ax12}
\det\left(\rho \,{\pmb {\rm I}}-{\pmb {\rm Y}}\right)=0\,.
\end{equation}
We observe some advantages of this system being compared to (\ref{ac-ax4}):
\begin{itemize}
    \item[(i)] The symmetric matrix ${\bf Y}$ has three real eigenvalues, but their sum is zero, ${\rm tr}\,{\pmb {\rm Y}}= \rho_1+\rho_2+\rho_3=0.$
    \item[(ii)] The shifted eigenvalues $\rho$ have the same multiplicity as the original eigenvalues $\lambda$.
    \item[(iii)] Furthermore, the solutions ${\pmb {\rm U}}$ for two eigenvalue problems (\ref{ac-ax4} and (\ref{ac-ax11}) are the same.
\end{itemize}

We can now write the characteristic equation (\ref{ac-ax6}) as a special cubic algebraic equation  
\begin{equation}\label{ac-ax13}
\rho^3+P\rho+Q=0.
\end{equation}
A reduced cubic equation or a cubic equation in the depressed form are the terms used in algebra to describe this kind of equation.  
The coefficients identified in (\ref{ac-ax13}) are expressed in terms of $\bY$ matrix invariants:
\begin{equation}\label{ac-ax14}
    P=-\frac12\,{\rm tr}\,{\pmb {\rm Y}}^2,\qquad Q=-{\rm det}\,{\pmb {\rm Y}}=-\frac 13\,{\rm tr}\,{\pmb {\rm Y}}^3\,.
\end{equation}

Observe that the coefficient $P=P({\bf n})$ is non-positive for any direction the vector ${\bf n}$ takes, whereas the coefficient $Q=Q({\bf n})$ can be of any sign. Using the traceless property ${\rm tr}\,{\pmb {\rm Y}}=0$  we obtain from (\ref{ac-ax8}) {\it the necessary and sufficient condition} for double axes to exist in the direction of ${\bf n}$
\begin{equation}\label{ac-ax15}
-\Delta=4P^3+27Q^2=0.
\end{equation}
Consequently, we have
\begin{prop} {\underline{\it Discriminant condition}}

An elastic medium with the reduced acoustic tensor $\bY(\bf n)$ has am acoustic axis along the direction of vector ${\bf n}$ if and only if 
\begin{equation}\label{ac-ax16}
\boxed{2({\rm det}\, {\pmb {\rm Y}})^2=\left(\frac 13\,{\rm tr}\, {\pmb {\rm Y}}^2\right)^3}\quad \Longleftrightarrow\qquad 
\boxed{ 6\,\left({\rm tr}\,{\pmb {\rm Y}}^3\right)^2=\left({\rm tr}\, {\pmb {\rm Y}}^2\right)^3}
\end{equation}
\end{prop}

\hspace{1.5cm}

The latter version in (\ref{ac-ax16}) follows from the former one by the Cayley–Hamilton theorem. Since $\bY$ is a quadratic homogeneous polynomial in the vector $\bf n$, both sides of Eq. (\ref{ac-ax16}) are homogeneous 12-th order polynamials of $\bf n$. The equations are preserved under rescaling of the vector $\bf n$ or altering its direction. 

The expression of the eigenvalues of matrix $\bY$ is a further useful result that can be obtained from Eq. (\ref{ac-ax13}); see \cite{It-axes1}. Let us denote 
\begin{equation}\label{ac-ax16x}
   \sigma =-3\,\frac{{\rm det}\,\bY}{{\rm tr}\,\bY^2 }.
\end{equation}
We will refer to this scalar as the {\it speed parameter}.  
Then the eigenvalues of the matrix $\bY$ in the direction of the acoustic axe are defined solely by the parameter $\sigma$,
\begin{equation}\label{ac-ax17}
\rho_1=\rho_2=\sigma, \qquad \rho_3=-2\sigma.
\end{equation}

Notice that the matrix $\bY$, being traceless, consistently possesses both positive and negative eigenvalues $\rho$. 
Being computed along the acoustic axes, the matrix $\bY$  takes on a unique diagonal form 
\begin{equation}\label{ac-ax17x}
     \bY= \sigma\begin{pmatrix}
1 &0&0 \\
0 &1&0\\
0 &0 &-2 
  \end{pmatrix}=\pm|\sigma|\begin{pmatrix}
1 &0&0 \\
0 &1&0\\
0 &0 &-2 
  \end{pmatrix}
\end{equation}
when the eigenvectors are utilized as the basis. Thus, for two alternative signs of the parameter $\sigma$, there exist two different forms; see Fig. 1. According to Norris \cite{Norris}, the direction with a negative double eigenvalue is referred to as the {\it prolate acoustic axe}, and the case with a positive double eigenvalue is called as the {\it oblate acoustic axe}. A {\it spherical axe} is the case when there is a triple eigenvalue, i.e., $\sigma=0$. Then, the reduced acoustic tensor vanishes identically. 

Since ${{\rm tr}\,\bY^2 }$ is non-negative, we have from (\ref{ac-ax16x}) an invariant  description of these three cases: 
\begin{eqnarray}
\sigma<0 \quad\Longleftrightarrow   \quad   {{\rm det}\,\bY}>0&&\qquad {\rm prolate\, acoustic \,axe}\\
 \sigma>0 \quad\Longleftrightarrow   \quad        {{\rm det}\,\bY}<0&&\qquad {\rm oblate\, acoustic \,axe}\\
\!\!\!\!\!\!\!\!\!\!\!\!\!\!\!\! \sigma=0 \quad\Longleftrightarrow   \quad  \quad  \,\,\, \bY=0&&\qquad {\rm spherical\, acoustic \,axe}
\end{eqnarray}
 \begin{figure}[H]\label{Fig1}
    \centering
\includegraphics[width=0.45\textwidth]{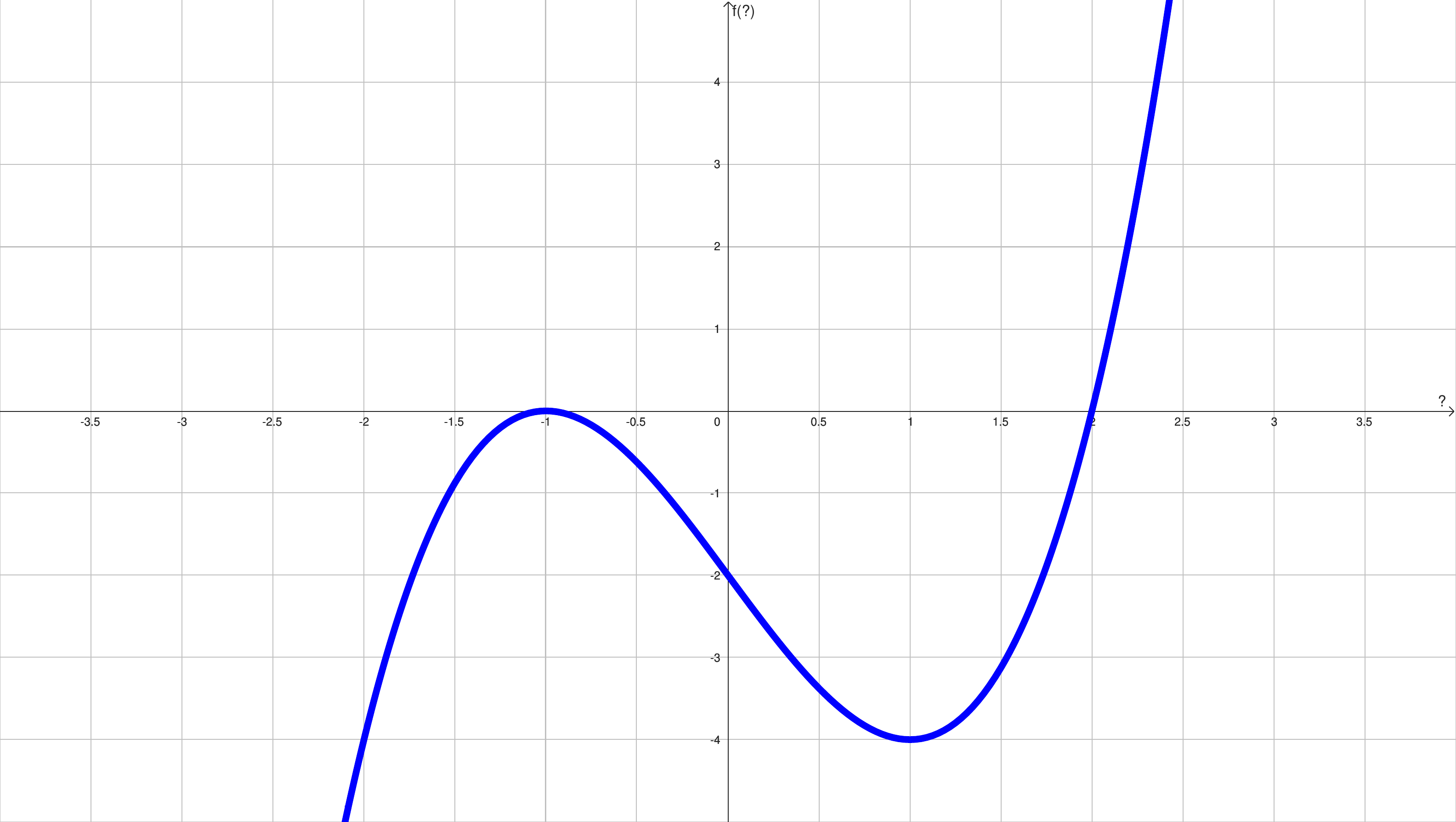}\qquad
\includegraphics[width=0.45\textwidth]{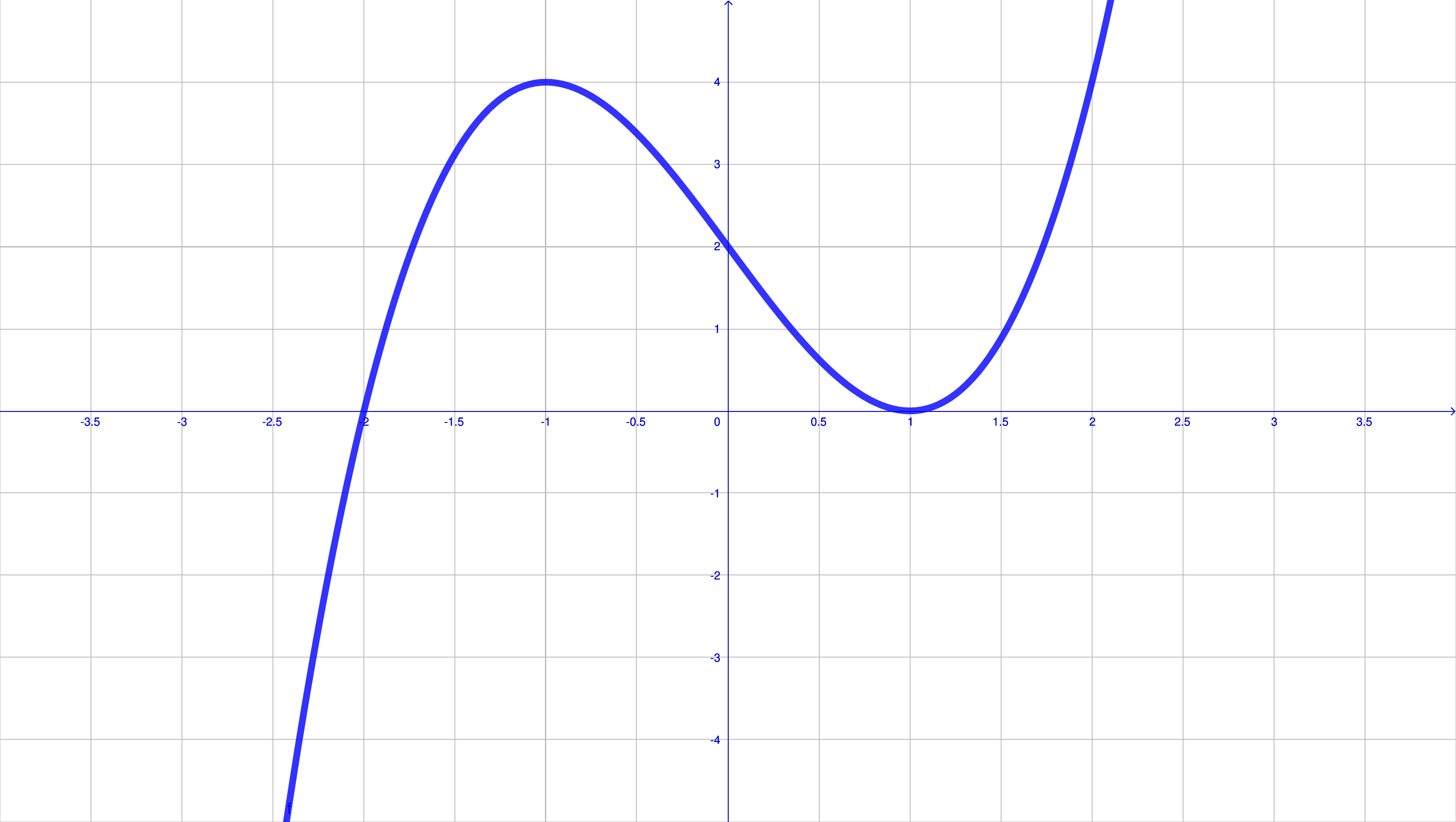}
\caption{Illustration of the characteristic polynomial in the case of the prolate and oblate acoustic axes, respectively.
}
\end{figure}

It is convenient to write the invariants of the matrix $\bY=\bY({\bf n})$ in terms of the parameter $\sigma$. Along the acoustic axe $\bf n$, we have
\begin{equation}
    {{\rm tr}\,\bY^2 }=6\sigma^2,\qquad {{\rm det}\,\bY}=-2\sigma^3.
\end{equation}
In this direction, the phase speeds are expressed for the double eigenvalue as 
\begin{equation}
    v^2=\gamma+\sigma=\frac 13 \,{\rm tr}\,\,{\pmb \Gamma}-3\,\frac{{\rm det}\,\bY}{{\rm tr}\,\bY^2 },
\end{equation}
and for the single eigenvalue as
\begin{equation}
    v^2=\gamma-2\sigma=\frac 13 \,{\rm tr}\,\,{\pmb \Gamma}+6\,\frac{{\rm det}\,\bY}{{\rm tr}\,\bY^2 }.
\end{equation}
Both these expressions must be positive. These requirements yield the following bounds on the speed parameter 
\begin{equation}\label{bound}
    \frac \gamma 2>\sigma>-\gamma.
\end{equation}
These inequalities provide the relationships between the elements of the elasticity tensor when the acoustic axes have already been computed. Subsequently, the inequalities among the distinct constituents of $C^{ijkl}$ can be obtained. When the specific examples of the elasticity materials are taken into consideration, we will revisit this topic.  
\section{Acoustic axis---adjoint matrix approach}
While the acoustic axes vector $\bf n$ is described by the compact relation (\ref{ac-ax16}), this 12th-order equation is too complex for actual computations. 
In this section, we present several possible versions of the acoustic axis condition. In the literature, these forms are expressed in terms of the complete Christoffel tensor $\bG$. Here, we reformulate them using solely the reduced Christoffel tensor $\bY$ as a building block. Additionally, we demonstrate how these conditions come from the same invariant origin, which we refer to as the  {\it adjoint matrix approach.}  
\subsection{Invariant conditions } 
The following algebraic observation leads to an alternate invariant description of the acoustic axes: 
\begin{prop}
 The characteristic equation 
\begin{equation}
    {\rm det}(\rho\, \bI-\bY)=0
\end{equation}
has a non-zero double root $\rho_1=\rho_2=\sigma$  if and only if the matrix 
\begin{equation}
   \bZ= \sigma\,\bI-\bY
\end{equation}
is of the rank one. 
\end{prop}
\begin{proof}
An equation employing the adjoint matrix (the matrix of cofactors) can be used to express the rank-one requirement. It is given by 
\begin{equation}\label{adj1}
   {\rm adj}\,\bZ= {\rm adj}\,(\sigma\, \bI-\bY)=0. 
\end{equation}              
Therefore, we have to support two claims:
 
 {\it Necessity:} Assuming that the characteristic equation has a double root $\rho_1=\rho_2=\sigma$, we will demonstrate that $\bZ$ must have the rank of one. 
 
 Using orthogonal transformations the symmetric matrix $\bY$ can be transformed into the diagonal form, where the eigenvalues are arranged along the diagonal. Due to the traceless property, the matrix $\bY$ takes the form $\bY={\rm diag}(\sigma,\sigma,-2\sigma)$ along the acoustic axes. Next, the diagonal form of the matrix $\bZ$ is shown as $\bZ={\rm diag}(0,0,2\sigma)$. Eq. (\ref{adj1}) so satisfied identically. This statement holds in any coordinate system, even though it has only been demonstrated in one specific coordinate system (where the matrices $\bY$ and $\bZ$ are symmetric). This is due to the fact that coordinate transformations do not affect the matrix rank. 
 
 {\it Sufficiency:} 
Assuming that the equation ${\rm adj}\,(\rho\, \bI-\bY)= 0$ holds for a certain real parameter $\rho$,  we prove that $\rho$ is a double eigenvalue of the symmetric matrix $\bY$.

Let us start with the adjoint ${\rm adj}\, \bZ$ of an arbitrary $(3\times 3)$-matrix $\bZ$ expanded as 
 \begin{equation}
 {\rm adj}\, \bZ= \frac1{2!}\epsilon_{ii_1i_2}\epsilon_{jj_1j_2} Z^{i_1j_1}Z^{i_2j_2},
\end{equation}
where $\epsilon_{ijk}$ is Levi-Civita's permutation pseudotensor. 
Recall the conventional identity for the permutation pseudotensor in three-dimensional space
\begin{equation}\label{ident01}
\epsilon_{ii_1i_2}\epsilon_{jj_1j_2}=\begin{vmatrix}
g_{ij} & g_{ij_1} &g_{ij_2}\\
g_{i_1j} & g_{i_1j_1} & g_{i_1j_2}\\
g_{i_2j} & g_{i_2j_1} & g_{i_2j_2}
\end{vmatrix}.
\end{equation}
Successive contractions of this product with the metric tensor are expressed as 
\begin{equation}\label{ident02}
g^{i_2j_2}\epsilon_{ii_1i_2}\epsilon_{jj_1j_2}=g_{ij}g_{i_1j_1}-g_{ij_1}g_{i_1j},\qquad  g^{i_1j_1}g^{i_2j_2}\epsilon_{ii_1i_2}\epsilon_{jj_1j_2}= 2g_{ij}\,,
\end{equation} 
and 
\begin{equation}\label{ident03}
g^{ij}g^{i_1j_1}g^{i_2j_2}\epsilon_{ii_1i_2}\epsilon_{jj_1j_2}=6.
\end{equation}    
 The adjoint matrix expression (\ref{adj1}) is then computed as follows
 \begin{eqnarray}\label{adj-cond0}
     {\rm adj}\,(\rho\, \bI-\bY)&=&\frac1{2!}\epsilon_{ii_1i_2}\epsilon_{jj_1j_2} 
 (\rho g^{i_1j_1}-Y^{i_1j_1})(\rho g^{i_2j_2}-Y^{i_2j_2})\nonumber\\
    &=& \rho^2g_{ij}-\rho Y^{i_1j_1}(g_{ij}g_{i_1j_1}-g_{ij_1}g_{i_1j})+({\rm adj}\,\bY)_{ij}\nonumber\\
    &=& \rho^2g_{ij}+\rho Y^{i_1j_1}g_{ij_1}g_{i_1j}+({\rm adj}\,\bY)_{ij}.
 \end{eqnarray}
 We used here the trace relation $Y^{i_1j_1}g_{i_1j_1}=0$. Applying the lower-indexed (contavariant) version of the tensor $Y_{ij}=Y^{i_1j_1}g_{ij_1}g_{i_1j}$ we obtain the component-wise form of the equation ${\rm adj}\,(\rho\, \bI-\bY)= 0$ 
 \begin{equation}\label{adj-cond}
     \rho^2g_{ij}+\rho Y_{ij}+({\rm adj}\,\bY)_{ij}=0.
 \end{equation}
The trace of this matrix equation yields one scalar equation
\begin{equation}\label{adj5}
3\rho^2+{\rm tr}({\rm adj}\,\bY)=0.
\end{equation}
Since for an arbitrary third-order  traceless matrix $\bY$,
\begin{equation}
    {\rm tr}({\rm adj}\,\bY)=-\frac 12\, {\rm tr}\,\bY^2,
\end{equation}
we can rewrite (\ref{adj5}) as
\begin{equation}\label{adj7}
\rho^2=\frac 16\, {\rm tr}\,\bY^2.
\end{equation}
Multiplying both sides of Eq.(\ref{adj-cond}) by $Y^{ij}$ and using the traceless property we get another scalar equation
\begin{equation}\label{adj8}
    \rho \,{\rm tr}\,\bY^2+3{\rm det}\,\bY=0.
\end{equation}
Thus, for ${\rm tr}\,\bY^2\ne 0$, we have
\begin{equation}\label{adj9}
    \rho=-3\,\frac {{\rm det}\,\bY}{{\rm tr}\,\bY^2}. 
\end{equation}
Comparing  (\ref{adj7}) and (\ref{adj9}), we obtain 
\begin{equation}
    \frac 16\, {\rm tr}\,\bY^2=\left(-3\,\frac {{\rm det}\,\bY}{{\rm tr}\,\bY^2}\right)^2.
\end{equation}
Then the cubic discriminant vanishes:
\begin{equation}
    \Delta=27({\rm det}\,\bY)^2-\frac 12 ({\rm tr}\,\bY^2)^3=0. 
\end{equation}
This means that the characteristic equation ${\rm det}(\rho\, \bI-\bY)=0$ has at least two repeating eigenvalues.
\end{proof}

During this proof, we derived an alternative form of the acoustic axes condition:
\begin{prop}{\underline{\it Adjoint condition:}}

The acoustic axis lies along the direction of the vector ${\bf n}$ if and only if the reduced acoustic tensor $\bY=\bY({\bf n})$ satisfies 
   \begin{equation}\label{adj-ax-cond}
\boxed{\sigma^2g_{ij}+\sigma Y_{ij}+({\rm adj}\,\bY)_{ij}=0}
\end{equation} 
for a certain scalar $\sigma$.
\end{prop}
Note that since this condition is based on tensors, it is invariant under coordinate transformations. 
Eq.(\ref{adj-ax-cond}) is a system of six equations for six variables: five components of the traceless symmetric matrix $\bY$ and a scalar parameter $\sigma$. We may extract the trace and traceless parts here. Then we have a system consisting of a scalar equation and a tensor equation with five components:
\begin{equation}
    \left\{ \begin{array}{l}
6\sigma^2= {\rm tr}\,\bY^2;\\
6\sigma Y_{ij}=-6({\rm adj}\,\bY)_{ij}-({\rm tr}\,\bY^2) g_{ij}
.\end{array} \right.
\end{equation}
These relations are quadratic polynomial equations in the components of the reduced acoustic tensor $\bY$, and therefore they are fourth-order equations in the propagation vector ${\bf n}$.
 
\subsection{Khatkevich  conditions}
For practical reasons, it can sometimes be useful to deal with non-invariant forms of the acoustic axes condition (\ref{adj-ax-cond}). Let us examine the non-diagonal components found in Eq. (\ref{adj-ax-cond}). Then we have
 \begin{equation}\label{Khat1}
     ({\rm adj}\,\bY)_{ij}=-\sigma Y_{ij},\qquad i\ne j.
 \end{equation}
In the Cartesian coordinates with the metric tensor $g_{ij}={\rm diag}(1,1,1)$, this system is presented  as 
 \begin{equation}\label{Khat2}
    \begin{cases}
     ({\rm adj}\,\bY)_{12}=-\sigma Y_{12}\\
      ({\rm adj}\,\bY)_{13}=-\sigma Y_{13}\\
      ({\rm adj}\,\bY)_{23}=-\sigma Y_{23}
    \end{cases}   \qquad \Longleftrightarrow\qquad
    \begin{cases}
     Y_{11}Y_{23}-Y_{12}Y_{13}=\sigma Y_{23}\\
     Y_{22}Y_{13}-Y_{12}Y_{23}=\sigma Y_{13}\\
     Y_{33}Y_{12}-Y_{23}Y_{13}=\sigma Y_{12}.\\
    \end{cases}   
\end{equation}
In general, we have here three independent fourth-order equations for two components of the unit vector ${\bf n}$  and the scalar parameter $\sigma$.  
 To eliminate the $\sigma$-terms we can utilize an assumption
\begin{equation}\label{assump1}
Y_{12}Y_{13}Y_{23}\ne 0.
 \end{equation}
Then, by combining the first two equations multiplying by $Y_{13}$ and $Y_{23}$, respectively, we obtain
 \begin{equation}\label{Khat3}
R_1=(Y_{11}-Y_{22})Y_{13}Y_{23}-Y_{12}(Y^2_{13}-Y^2_{23})=0.
 \end{equation}
Similarly, by combining the first and the third equations in (\ref{Khat2}) we obtain 
\begin{equation}\label{Khat4}
R_2=(Y_{11}-Y_{33})Y_{12}Y_{23}-Y_{13}(Y^2_{12}-Y^2_{23})=0.    \end{equation}
We observe that the system $(R_1,R_2)$ is completely equivalent to the original Khatkevich's system \cite{Khatkevich}. Indeed,  the off-diagonal elements of the matrices $\bY$ and $\bG$ coincide. As for the diagonal elements, they appear in Eqs. (\ref{Khat3}--\ref{Khat4}) only in a difference form. Then $Y_{11}-Y_{22}=\Gamma_{11}-\Gamma_{22}$ and $Y_{11}-Y_{33}=\Gamma_{11}-\Gamma_{33}$. Certainly, the elements of the matrices $\bY$ (or $\bG$) cannot be derived from the system $(R_1,R_2)$ uniquely. However, in terms of two independent components of the unit covector ${\bf n}$, we have here two equations of the sixth order. So the system turns out to be applicable for the actual calculations of the wave covector ${\bf n}$ at least for the non-exceptional cases. 

When at least one off-diagonal member of the matrix $Y^{ij}$ is zero, the system (\ref{Khat2}), like in the original Khatkevich conditions, gives trivial identities. In this instance, there is no information regarding the acoustic axis in the conditions (\ref{Khat3}--\ref{Khat4}).

\subsection{Alshits\&Lothe conditions}
When calculating acoustic axes in the non-exceptional case (\ref{assump1}), Khatkevich's conditions are helpful. They need to be expanded out to accommodate all possible axes. Furthermore,  these conditions are not invariant under the rotations of the basis. Alshits and Lothe \cite{Alshits} suggested adding five more equations to the system $(R_1,R_2)$ in order to address these issues. Let us demonstrate how the adjoint condition (\ref{adj-ax-cond}) yields these extra equations. Combining the second and third equations in (\ref{Khat2}) we obtain the $R_3$-equation from \cite{Alshits}
\begin{equation}
R_3=\left(Y_{22}-Y_{33}\right)Y_{12}Y_{13}-Y_{23}\left(Y^2_{12}-Y^2_{13}\right)=0.
\end{equation}
Additionally, we have three equations corresponding the the diagonal of the system (\ref{adj-ax-cond})
 \begin{equation}\label{Al2}
    \begin{cases}
    6\sigma Y_{11}=-6({\rm adj}\,\bY)_{11}-{\rm tr}\,\bY^2\\
    6\sigma Y_{22}=-6({\rm adj}\,\bY)_{22}-{\rm tr}\,\bY^2\\
     6\sigma Y_{33}=-6({\rm adj}\,\bY)_{33}-{\rm tr}\,\bY^2
    \end{cases}   \qquad \Longrightarrow\qquad
    \begin{cases}
\sigma (Y_{11}-Y_{22})=-({\rm adj}\,\bY)_{11}+({\rm adj}\,\bY)_{22}\\
\sigma (Y_{22}-Y_{33})=-({\rm adj}\,\bY)_{22}+({\rm adj}\,\bY)_{33}\\
\sigma (Y_{33}-Y_{11})=-({\rm adj}\,\bY)_{33}+({\rm adj}\,\bY)_{11}.
    \end{cases}   
\end{equation}
Calculating the elements of the adjoint matrix we have
\begin{equation}\label{Al2a}
    \begin{cases}
\sigma (Y_{11}-Y_{22})=(Y_{11}-Y_{22})Y_{33}+Y^2_{23}-Y^2_{13}\\
\sigma (Y_{22}-Y_{33})=(Y_{22}-Y_{33})Y_{11}+Y^2_{13}-Y^2_{12}\\
\sigma (Y_{33}-Y_{11})=(Y_{33}-Y_{11})Y_{22}+Y^2_{12}-Y^2_{23}.
\end{cases}
    \end{equation}
By utilizing algebraic manipulations, we remove the $\sigma$-terms to derive the remaining Alshits\&Lothe conditions.  For example, combining the third equation in (\ref{Al2a}) multiplied by $Y_{23}$ with the  first equation of (\ref{Khat2}) multiplied by $(Y_{33}-Y_{11})$ we get 
\begin{equation}\label{Al2a4}
R_4=(Y_{11}-Y_{22})(Y_{11}-Y_{33})Y_{23}-(Y_{11}-Y_{33})Y_{12}Y_{13}+Y_{23}(Y^2_{12}-Y^2_{23})=0.
    \end{equation}
Similarly combining the equations from (\ref{Khat2}) and (\ref{Al2}) and illuminating the $\sigma$-terms we obtain 
\begin{equation}\label{Al2b}
R_5=(Y_{22}-Y_{11})(Y_{22}-Y_{33})Y_{13}-(Y_{22}-Y_{33})Y_{12}Y_{23}+Y_{13}(Y^2_{12}-Y^2_{13})=0,
    \end{equation}
and
\begin{equation}\label{Al2c}
R_6=(Y_{33}-Y_{11})(Y_{33}-Y_{22})Y_{12}-(Y_{33}-Y_{22})Y_{13}Y_{23}+Y_{12}(Y^2_{13}-Y^2_{12})=0.
    \end{equation}
Finally, from the first pair of equations in (\ref{Al2}),   we obtain 
\begin{equation}\label{Al3}
R_7=(Y_{11}-Y_{22})(Y_{22}-Y_{33})(Y_{11}-Y_{33})+(Y_{22}-Y_{33})(Y^2_{13}-Y^2_{23})+(Y_{11}-Y_{22})(Y^2_{13}-Y^2_{12})=0
    \end{equation}
Even being written in terms of the matrix $Y_{ij}$ the system $(R_1,\cdots, R_7)$ includes only the differences of the diagonal elements. So it is completely equivalent to the original system of Alshits and Lothe written in terms of the matrix $\Gamma_{ij}$; see \cite{Alshits}. 

The invariance properties follow from the invariance of the tensor equation given in (\ref{adj-ax-cond}). Indeed when the coordinates transformed as $x^i\to x^{i'}$ equation (\ref{adj-ax-cond}) takes the form 
\begin{equation}\label{adj-ax-cond1}
\sigma^2g_{i'j'}+\sigma Y_{i'j'}+({\rm adj}\,\bY)_{i'j'}=0.
\end{equation}
The system $(R_1,\cdots, R_7)$ for the transformed components follows from this equation by the same procedure as above. 

Let us examine the case in which the assumption (\ref{assump1}) is not satisfied, that is, in which the matrix $\bY$ has at least one zero out-of-diagonal element. For example, suppose $Y_{23}=0$.  
Then, one more zero out-of-diagonal element is required, as may be inferred from (\ref{Khat3}). Consider now the case $Y_{13}=Y_{23}=0$. Then the equations $R_1=R_2=R_3=R_4=R_5=0$ hold identically. We are left with a pair
\begin{equation}
    \begin{cases}
        R_6=\Big((Y_{33}-Y_{11})(Y_{33}-Y_{22})-Y^2_{12}\Big)Y_{12}=0\\
        R_7=(Y_{11}-Y_{22})\Big((Y_{22}-Y_{33})(Y_{11}-Y_{33})-Y^2_{12}\Big)=0
    \end{cases}
\end{equation}
Two possible cases can be identified here:
\begin{itemize}
    \item If $Y_{12}=0$, we are left with the matrix of the scalar form $Y_{ij}=\sigma\,{\rm diag}(1,1,-2)$. The parameter $\sigma$ is derived when the expressions for $Y_{11}$ and $Y_{22}$ are substituted.
    \item If $Y_{12}\ne 0$, we are left with one fourth-order equation  $Y^2_{12}=(Y_{22}-Y_{33})(Y_{11}-Y_{33})$ for the components of the vector ${\bf n}$. It  must be solved together with the second-order equations $Y_{13}=Y_{23}=0$. 
\end{itemize}
In the follow-up, we will use examples of real crystals to discuss these exceptional cases. 

\subsection{Norris tensor condition}
An alternative invariant condition for the existence of the acoustic axe in a given direction ${\bf n}$ was proposed by Norris in \cite{Norris}. One starts with an arbitrary vector ${\bf m}$ and consider the  scalar triple product of three vectors, ${\bf m},{\bG}{\bf m}$, and ${\bG}^2{\bf m}$. Specifically, Norris demonstrated by the use of the Gibbs decomposition that a Christoffel tensor $\bG$ is uniaxial (has two equal eigenvalues) if and only if the triple product of these three vectors vanishes, 
\begin{equation}\label{Nor1}
    [{\bf m},{\bG}{\bf m},{\bG}^2{\bf m}]=0
\end{equation}
for an arbitrary vector ${\bf m}$. Recall the expression for the triple product (scalar-vector product) of three vectors: $[{\bf a},{\bf b},{\bf c}]=
{\bf a}\cdot({\bf b}\times {\bf c})$. Since this product  represents the volume of the parallelepiped spanned by the vectors ${\bf a},{\bf b},{\bf c}$, it vanishes if and only if the three vectors are co-planar.  
Equation (\ref{Nor1}) must be satisfied for any vector ${\bf m}$ when the matrix $\bG=\bG({\bf n})$ is considered  along the acoustic axis ${\bf n}$. Next,  Norris constructed a third-order tensor $\phi^{ijk}$  that is independent of the arbitrary vector ${\bf m}$. It is defined implicitly from the equation  
\begin{equation}\label{Nor2}
    [{\bf m},{\bG}{\bf m},{\bG}^2{\bf m}]=\phi^{ijk}m_im_jm_k.
    \end{equation}
    Note that the {\it Norris tensor} $\phi^{ijk}$ is fully symmetric and traceless as a result of this definition, \cite{Norris}.
Formally, these facts might be expressed as 
\begin{equation}\label{Nor3a}
    \phi^{ijk}=\phi^{(ijk)}, \qquad \phi^{ijk}g_{jk}=0.
\end{equation}
The proof can be given as follows.  The definition of the Norris tensor (\ref{Nor2}) is provided by a product with the fully symmetric combination $m_im_jm_k$, which explains why the symmetry equation in (\ref{Nor3a}) holds.  For the traceless equation, considering the tensor $\bG$ relative to a particular basis suffices for proof. By implementing the eigenvector basis and utilizing the diagonal form of the matrix $\bG={\rm diag}(\lambda_1,\lambda_2,\lambda_3)$, we can calculate
\begin{equation}\label{Nor4}
    [{\bf m},{\bG}{\bf m},{\bG}^2{\bf m}]=\begin{vmatrix}
m_1 & \lambda_1m_1 & \lambda^2_1m_1\\
m_2 & \lambda_2m_2 & \lambda^2_1m_2\\
m_3 & \lambda_3m_3 & \lambda^2_3m_3
\end{vmatrix}=m_1m_2m_3(\lambda_1-\lambda_2)(\lambda_2-\lambda_3)(\lambda_3-\lambda_1). 
\end{equation}
The identities (\ref{Nor3a}) are thus valid since the Norris tensor in this specific basis has just one non-zero component (up to transition of the indices):
\begin{equation}\label{Nor5}
   \phi_{123}= (\lambda_1-\lambda_2)(\lambda_2-\lambda_3)(\lambda_3-\lambda_1).
\end{equation}
Owing to the invariance of the conditions (\ref{Nor3a}), they  remain true in any basis. The identities (\ref{Nor3a}) reduce the tensor $\phi^{ijk}$ from twenty-seven to seven independent components. Furthermore, the acoustic axis requires that at least two of the eigenvalues be identical, which is indicated by the invariant relation 
\begin{equation}\label{Nor5a}
    \phi^{ijk}=0.
\end{equation}

It is feasible to infer that the Norris tensor must depend on just five distinct components of the reduced Christoffel tensor $\bY$ by comparing this requirement to the other versions of the acoustic axis criteria. 
Then for an arbitrary vector ${\bf m}$, a similar condition  
\begin{equation}\label{Nor6}
    [{\bf m},{\bY}{\bf m},{\bY}^2{\bf m}]=0
\end{equation}
must hold along the acoustic axis. Indeed, substituting into (\ref{Nor1}) the expressions 
\begin{equation}
    \bG=\bY+\gamma\bI,\qquad \bG^2=\bY^2+2\gamma\bY+\gamma^2\bI
\end{equation}
we obtain
\begin{equation}\label{Nor3aa}
    [{\bf m},{\bG}{\bf m},{\bG}^2{\bf m}]=[{\bf m},{\bY}{\bf m}+\gamma{\bf m},{\bY}^2{\bf m}+2\gamma{\bY}{\bf m}+\gamma^2{\bf m}]=[{\bf m},{\bY}{\bf m},{\bY}^2{\bf m}].
\end{equation}
Consequently the definition of the Norris tensor (\ref{Nor2}) can be rewritten in terms of the reduced Christoffel tensor 
\begin{equation}\label{Nor4aa}
    [{\bf m},{\bY}{\bf m},{\bY}^2{\bf m}]=\phi^{ijk}m_im_jm_k.
    \end{equation}
with the same  tensor $\phi^{ijk}$.  
    
To obtain the explicit representation of the tensor $\phi^{ijk}$ in terms of the acoustic tensor, we utilize the standard expression of the scalar triple product using the permutation pseudo-tensor, 
\begin{equation}
    [{\bf a},{\bf b},{\bf c}]=
\varepsilon^{ijk}a_ib_jc_k.
\end{equation}
We write the components of the vectors in the form
\begin{equation}\label{Nor5aa}
 ({\bY}{\bf m})_r=Y_r{}^jm_j,\qquad  ({\bY}^2{\bf m})_s=  Y_s{}^pY_p{}^km_k
\end{equation}
Then
\begin{equation}
[{\bf m},{\bY}{\bf m},{\bY}^2{\bf m}]= \left(\varepsilon^{irs}Y_r{}^jY_s{}^pY_p{}^k\right)m_im_jm_k.
\end{equation}
Since this equation holds for an arbitrary vector ${\bf m}$, the third-order Norris tensor is expressed as 
\begin{equation}
\phi^{ijk}=
\underset{(i,j,k)} {\rm sym}(\varepsilon^{irs}Y_r{}^jY_s{}^pY_p{}^k),
\end{equation} 
where the ${\rm sym}$-operator on the right-hand side denotes the symmetrization of the expression with respect to the indices $(i,j,k)$. 

Similarly to the Khatkevich and Alshits\&Lothe equations, the Norris condition is expressed as a third-order polynomial in terms of the tensor $\bY$. However, these expressions are not immediately related to one another.
 In the section that follows, we shall address this matter again. 
\section{Minimal polynomial approach}
In this section, we propose an alternative approach to the acoustic axes problem. It is based on the properties of the minimal polynomial of a matrix. The minimal polynomial is an invariant construction then it can be expanded for any second-order tensor.
\subsection{Minimal polynomial at acoustic axes}
According to the well-known Cayley-Hamilton theorem, each square matrix fulfills its own characteristic equation.  
In particular, utilizing the characteristic equation $f(\rho)=\rho^3+P\rho+Q=0$ for the reduced acoustic matrix $\bY=\bY(\bf n)$ we obtain the matrix equation
\begin{equation}\label{min1-0}
f(\bY)=\bY^3+P\bY+Q\bI=0,
\end{equation}
that holds identically for arbitrary directions of the propagation vector $\bf n$. 
This is a third-order polynomial equation for the matrix $\bY$. Other polynomials of order greater than three that vanish for the matrix $\bY$ are certainly possible to construct. There is, however, a single polynomial of a lowest order of this type. 

So the {\it minimal polynomial} $m(\rho)$ of the matrix $\bY$ is defined as the {\it smallest order} monic polynomial (with the leading coefficient equal one) that vanishes for the matrix $m(\bY)=0$. 
Let us recall some basic algebraic facts regarding the minimal polynomials: 

$\bullet$ For each given matrix, the minimal polynomial is uniquely definite.

$\bullet$ The minimal polynomial coincides with the characteristic polynomial of a symmetric $(3\times 3)$-matrix with three distinct eigenvalues, $\rho_1,\rho_2,\rho_3$. 
Thus if three roots of the polynomial in Eq.(\ref{min1-0}) are distinct the minimal polynomial is expanded as 
\begin{equation}\label{min1-1}
   m(\rho)= (\rho-\rho_1)(\rho-\rho_2)(\rho-\rho_3).
\end{equation}
The corresponding matrix equation can be written in the form
\begin{equation}
    m(\bY)= (\bY-\rho_1\bI)(\bY-\rho_2\bI)(\bY-\rho_3\bI)=0,
\end{equation}
that is equivalent to (\ref{min1-0}). 
Note that this relation holds identically for an arbitrary direction ${\bf n}$. 

\vspace{0.2cm}

$\bullet$ The characteristic polynomial in the case of a double eigenvalue $\rho_1=\rho_2$ is of the third order
\begin{equation}
   f(\rho)= (\rho-\rho_1)^2(\rho-\rho_3),
\end{equation}
whilst the minimal polynomial turns out to be of the second order,  
\begin{equation}\label{min1-2}
    m(\rho)=(\rho-\rho_1)(\rho-\rho_3). 
\end{equation}
In this case, the matrix equation takes the form 
\begin{equation}\label{min1-3}
    m(\bY)=(\bY-\rho_1\bI)(\bY-\rho_3\bI)=0. 
\end{equation}
This equation holds only for the directions of the acoustic axes. For the other  directions, the expression in the left-hand side is not equal zero. 
Consequently, the minimal polynomial is able to distinguish the acoustic axes from other wave propagation directions: along the acoustic axe, it is of the second order, while along other directions it is of the third order. 

We will apply now the traceless property of the matrix $\bY$. Then along the acoustic axis the eigenvalues can be presented by
\begin{equation}\label{min1-4}
\rho_1=\rho_2=\sigma,\qquad \rho_3=-2\sigma.
\end{equation}
Consequently, 
\begin{prop}{\underline{\it Minimal polynomial condition}}

    The reduced acoustic tensor $\bY=\bY({\bf n})$ has an acoustic axis in the direction ${\bf n}$ if and only if the equation 
    \begin{equation}\label{min1-6}
    \boxed{\bY^2+\sigma\bY-2\sigma^2\bI=0}
    \end{equation}
   holds for some scalar $\sigma$.

\end{prop}

\hspace{1cm}

This equation can be seen as a system of six independent equations for six independent variables: five components of the traceless symmetric matrix ${\bf Y}$ and one component of the speed parameter $\sigma.$ 
Let us recall the expressions  for the parameter $\sigma$:  
\begin{equation}\label{min1-5}
    \sigma=-3\frac{{\rm det}\,{\pmb {\rm Y}}}{{\rm tr}\,{\pmb {\rm Y}}^2 },\qquad 
    \sigma^2=\frac 16 {{\rm tr}\,\bY^2 }.
\end{equation}

\subsection{Minimal polynomial vs adjoint}
The adjoint and minimal polynomial conditions are represented by quadratic equations in terms of the tensor $\bY$. Because the minimal polynomial is unique (up to the leading coefficient), these equations must be equivalent.
 Let us compare the minimal polynomial equation
\begin{equation}\label{min-vs1}
   \bY^2+\sigma\bY-2\sigma^2\bI=0.
\end{equation}
 and the adjoint matrix equation
 \begin{equation}\label{min-vs2}
   ({\rm adj}\,\bY)+\sigma\bY
   +\sigma^2\bI=0.
\end{equation}
In order to prove the eqivalence, we multiply both sides of Eq.(\ref{min-vs2}) by the matrix $\bY$. Then we obtain 
\begin{equation}\label{min-vs3}
  \bY({\rm adj}\,\bY)+\sigma\bY^2
   +\sigma^2\bY=0.
\end{equation}
In the first term, we recognize the determinant of the matrix. 
Consequently, we have 
\begin{equation}\label{min-vs4}
  \sigma\bY^2
   +\sigma^2\bY+ ({\rm det}\,\bY)\bI=0.
\end{equation}
Utilizing the diagonal form of the matrix $\bY={\rm diag}(\sigma,\sigma,-2\sigma)$, we obtain the representation ${\rm det}\,\bY=-2\sigma^3$. Notice that the latter equation holds in arbitrary coordinates due to its invariance. Substitution into (\ref{min-vs4}) yields 
\begin{equation}\label{min-vs5}
   \sigma(\bY^2
   +\sigma\bY-2\sigma^2\bI)=0.
\end{equation}
For $\sigma\ne 0$, it is equivalent to the minimal polynomial equation (\ref{min-vs1}).

\subsection{Component-wise representation}
In the components, Eq.(\ref{min1-6}) reads
\begin{equation}\label{min-comp-0}
    Y^{ij}Y_j{}^k+\sigma Y^{ik}-2\sigma^2 g^{ik}=0\,.
\end{equation}
This matrix equation can be seen as a system of six equations---three diagonal and three out-of-diagonal ones. The diagonal equations are obtained with the values of indices $i=k=1,2,3$.
\begin{equation}\label{min-comp-1}
    \begin{cases}
(Y_{11})^2+(Y_{12})^2+(Y_{13})^2+\sigma Y_{11}-2\sigma^2=0\\
(Y_{22})^2+(Y_{12})^2+(Y_{23})^2+\sigma Y_{22}-2\sigma^2=0\\
(Y_{33})^2+(Y_{13})^2+(Y_{23})^2+\sigma Y_{33}-2\sigma^2=0.\\
    \end{cases}
\end{equation}
The system of three out-of-diagonal  equations with the pairs of indices $(i,k)=(1,2),(1,3), (2,3)$ is written as
\begin{equation}\label{min-comp-2}
    \begin{cases}
Y_{11}Y_{12}+Y_{12}Y_{22}+Y_{13}Y_{32}+\sigma Y_{12}=0\\
Y_{11}Y_{13}+Y_{12}Y_{23}+Y_{13}Y_{33}+\sigma Y_{13}=0\\
Y_{12}Y_{13}+Y_{22}Y_{23}+Y_{23}Y_{33}+\sigma Y_{23}=0.
    \end{cases}
\end{equation}
We will show that system (\ref{min-comp-1}) almost always follows from the system (\ref{min-comp-2}). Using the traceless condition $Y_{11}+Y_{22}+Y_{33}=0$ we rewrite (\ref{min-comp-2}) as
\begin{equation}\label{min-comp-3}
    \begin{cases}
Y_{13}Y_{23}=Y_{12}(Y_{33}-\sigma) \\
Y_{12}Y_{23}=Y_{13}(Y_{22}-\sigma)\\
Y_{12}Y_{13}=Y_{23}(Y_{11}-\sigma).
    \end{cases}
\end{equation}
We assume $Y_{12}Y_{13}Y_{23}\ne 0$ and multiply in pairs the equations of this system to obtain
\begin{equation}\label{min-comp-4}
    \begin{cases}
(Y_{12})^2=(Y_{11}-\sigma)(Y_{22}-\sigma) \\
(Y_{13})^2=(Y_{11}-\sigma)(Y_{33}-\sigma)\\
(Y_{23})^2=(Y_{22}-\sigma)(Y_{33}-\sigma)
    \end{cases}
\end{equation}
Substituting these relations into (\ref{min-comp-1}) and using the traceless condition, we obtain the identities of the type $0=0$ that completes the proof. 

Observe that (\ref{min-comp-3}) can be rewritten as a system of two independent equations. Eliminating the parameter $\sigma$, we are left with the relations 
\begin{equation}
  \frac{Y_{13}Y_{23}}{Y_{12}}-Y_{33}=
  \frac{Y_{12}Y_{23}}{Y_{13}}-Y_{22} =
  \frac{Y_{12}Y_{13}}{Y_{23}}-Y_{11}\,.
\end{equation}
These two equations are equivalent to the original Khatkevich's system. This two-equation system suffices in the majority of ordinary situations, as the above analysis makes clear. 

Now, let us examine the  exceptional situations. For 
\begin{equation}
    Y_{12}Y_{13}Y_{23}=0,
\end{equation} 
systems (\ref{min-comp-1}---\ref{min-comp-3}) provides three possible cases:

(i) In the case $Y_{12}=Y_{13}=Y_{23}=0$, system (\ref{min-comp-2}) satisfies identically, while (\ref{min-comp-1}) takes the form 
\begin{equation}\label{min-comp-5}
    \begin{cases}
(Y_{11}+2\sigma)(Y_{11}-\sigma)=0\\
(Y_{22}+2\sigma)(Y_{22}-\sigma)=0\\
(Y_{33}+2\sigma)(Y_{33}-\sigma)=0\\
    \end{cases}
\end{equation}
Due to the traceless condition, $Y_{11}+Y_{22}+Y_{33}=0$, we are left here with only three different solutions:
\begin{equation}\label{min-comp-6}
    \bY= \sigma\begin{pmatrix}
-2 &0&0 \\
0 &1&0\\
0 &0 &1 
  \end{pmatrix}\,,\quad
      \bY= \sigma\begin{pmatrix}
1 &0&0 \\
0 &-2&0\\
0 &0 &1 
  \end{pmatrix}\,,\quad
      \bY= \sigma\begin{pmatrix}
1 &0&0 \\
0 &1&0\\
0 &0 &-2 
  \end{pmatrix}\,.
\end{equation}

(ii) Consider the case when only two out-of-diagonal terms vanish. Let, for instance,  $Y_{13}=Y_{23}=0$, but $Y_{12}\ne 0$. Then system (\ref{min-comp-3}) is equivalent to only one equation $Y_{33}=\sigma$ and system (\ref{min-comp-1}) gives 
 two equation
\begin{equation}
\begin{cases}
(Y_{11})^2+(Y_{12})^2+\sigma Y_{11}-2\sigma^2=0\\
(Y_{22})^2+(Y_{12})^2+\sigma Y_{22}-2\sigma^2=0,
\end{cases}
\end{equation}
that are dependent because of the trace equation $Y_{11}+Y_{22}+\sigma=0.$ Eliminating the parameter $\sigma$ we are left with the system of three equations
\begin{equation}
    \begin{cases}
        Y_{13}=Y_{23}=0\\
        (Y_{33}-Y_{11})(Y_{33}-Y_{22})=Y_{12}^2.
    \end{cases}
\end{equation}
This system is sufficient for determining two independent components of the  vector ${\bf n}.$ 
The interchanging of the indices yields two additional results:
\begin{equation}
    \begin{cases}
        Y_{12}=Y_{13}=0\\
        (Y_{11}-Y_{22})(Y_{11}-Y_{33})=Y_{23}^2
    \end{cases}\qquad {\rm and }\qquad 
    \begin{cases}
        Y_{12}=Y_{23}=0\\
        (Y_{22}-Y_{11})(Y_{22}-Y_{33})=Y_{13}^2.
    \end{cases}
\end{equation}

(iii) The case with only one zero non-diagonal component does not exist. Indeed, for  $Y_{23}=0$, Eq.(\ref{min-comp-6}) results in $Y_{12}Y_{13}=0$, i.e., at least one additional non-diagonal component is zero. Then we come back to the cases considered above. 
\subsection{Norris approach}
The Norris conditions follow from the equation  
\begin{equation}\label{Nor2aaa}
    [{\bf m},{\bY}{\bf m},{\bY}^2{\bf m}]=0
\end{equation}
that is assumed to hold for an arbitrary vector ${\bf m}$. 
We can establish the relation between this condition and the minimal polynomial equation. 
\begin{prop}
    The Norris tensor and the minimal polynomial vanish simultaneously
    \begin{equation}
        \phi^{ijk}=0\qquad \Longleftrightarrow\qquad  m(\bY)=0.
    \end{equation}
\end{prop}
\begin{proof}
    Let us assume that the minimal polynimial condition holds,
\begin{equation}
  m(\bY)=\bY^2+\sigma\bY-2\sigma^2\bI=0.  
\end{equation}
Multiplying both sides of this equation by an arbitrary vector ${\bf m}$ we obtain a linear relation between three vectors appearing into (\ref{Nor2aaa}) 
\begin{equation}\label{Nor3}
\bY^2 {\bf m}=-\sigma\bY {\bf m}+2\sigma^2{\bf m}.
\end{equation}
 The triple product of linearly related vectors in (\ref{Nor2}) necessary vanishes. 
 
 This condition is also sufficient because the acoustic axes are the only directions where the third-order minimal polynomial turns into the second-order one.  In all directions different from acoustic axes there are no quadratic polynomials vanishing for the tensor $\bY$, so the vectors are linearly independent. 

\end{proof}

\section{Polarization vector representation}
While the minimal polynomial condition is equivalent to the usual adjoint condition, it offers several advantages. 
In this section, we demonstrate that the minimal polynomial equation can be solved in terms of the matrix $\bY$. Then we can derive a second-order acoustic axes condition in this manner.
\subsection{Orthogonal matrix representation}
The minimal polynomial quadratic equation 
\begin{equation}\label{min1-6a}
    \bY^2+\sigma\bY-2\sigma^2\bI=0.
\end{equation}
can be rewritten as
\begin{equation}\label{min1-7}
    \left(\bY+\frac \sigma 2\bI\right)^2=\frac {9\sigma^2} 4\bI.
\end{equation}
The zero value of the parameter $\sigma$ corresponds to $\bY=0$. This is a case of a triple axis.  We assume that $\sigma\ne 0$ to handle the double axis. Let us denote
\begin{equation}\label{min1-8}
 \bZ= \frac 2{3\sigma}\bY+\frac 13\bI. 
\end{equation}
Following that, a very compact form 
\begin{equation}
    \bZ^2({\bf n})=\bI.
\end{equation}
is provided for equation  (\ref{min1-6a}) guaranteeing the existence of an acoustic axe in the direction ${\bf n}$. 
\subsection{Solution of the matrix equation}
To describe all acoustic axes we have to derive the most general solution to the equation $\bZ^2=\bI$. First, observe that due to the definition (\ref{min1-8}), we have ${\rm tr} \,\bZ=1$. Then the immediate solution $\bZ=\pm \bI$ does not fit the definition. Observe that since the matrix $\bZ$ is symmetric and involutory it is orthogonal. Then we can characterize the solutions by the well-known representations of the orthogonal matrices. As a result, we have:    
\begin{prop}
The general form of the solution to equation $\bZ^2=\bI$ for a symmetric third-order matrix $\bZ$ satisfying the relationship ${\rm tr} \,\bZ=1$ is given by
\begin{equation}\label{sol-Z}
    \bZ= \bI-2{\bf q}\otimes{\bf q},
\end{equation}
where ${\bf q}$ is an arbitrary unit vector.
\end{prop}
\begin{proof}
    In order to solve the equation $\bZ^2=\bI$, we first observe the defining properties of the matrix $\bZ$. This matrix is symmetric, $\bZ=\bZ^T$, and involutory, $\bZ^2=\bI$. Consequently, it is an orthogonal matrix, $\bZ\bZ^T=\bI$.   
Using the well-known Rodriguez formula we can express this orthogonal matrix by a mechanical rotation. In particular, 
\begin{equation}\label{min1-9}
    Z^{ij}=\a g^{ij}+\b q^i q^j +\g \varepsilon^{ijk}q_k,
\end{equation}
where the coefficients are parametrized by one parameter $\theta$
\begin{equation}\label{min1-10}
    \a=\cos\theta,\qquad \b=\pm(1\mp\cos\theta),\qquad \g=\sin\theta.
\end{equation}
Here $\theta$ describes an angle of a
rotation about the axis {\bf q}. The upper sign corresponds to the proper rotation while the lower sign is used for improper rotation (rotation together with a mirror reflection). We deduce that the matrix $\bZ$ can be symmetric if and only if $\g=\sin\theta=0$, i.e., $\cos\theta=\pm 1$.
Consequently, we have four possible forms:
\begin{equation}
    Z^{ij}=\pm g^{ij},\qquad Z^{ij}=\pm \left(g^{ij}-2q^iq^j\right).
\end{equation}
Requiring ${\rm tr}\,\bZ=1$ we are left with one possible form
\begin{equation}
Z^{ij}=g^{ij}-2q^iq^j,\qquad {\rm or}\qquad {\bZ}=\bI-2{\bf q}\otimes{\bf q}.
\end{equation}
\end{proof}
\subsection{Acoustic axes condition in a vector form}
Comparing equations (\ref{min1-8}) and (\ref{min1-9}) we obtain the most general solution of the minimal polynomial quadratic equation (\ref{min1-6a})
\begin{equation}\label{min1-11}
  \bY({\bf n})= \sigma(\bI-3{\bf q}\otimes{\bf q}).
\end{equation}
or, in the components, 
\begin{equation}\label{min1-12}
 Y^{ij}({\bf n})=\sigma(g^{ij}- 3q^i q^j).
\end{equation}
Observe that these equations are covariant---they hold in an arbitrary coordinate system. The tensor in the left-hand side of Eq.(\ref{min1-11}) is traceless as it is required by the definition of the tensor $\bY$. 

In our derivation, $\bf q$ appears as an arbitrary vector. Nonetheless, a definite physical interpretation can be assigned to it.
 Multiplying both sides of Eq.(\ref{min1-12}) by ${\bf q}$ we obtain
\begin{equation}
    Y^{ij}q_j=-2\sigma q^j.
\end{equation}
Thus ${\bf q}$ is the eigenvector of the matrix $Y^{ij}$ corresponding to its single eigenvalue $\rho=-2\sigma$. Consequently, when the system (\ref{min1-12}) is solved we obtain the directions of the acoustic axes and the polarization of the wave corresponding to the single eigenvalue. It is important to keep in mind that, for a given direction vector ${\bf n}$, the speed scalar parameter $\sigma$ and  the polarization vector ${\bf q}$ (up to its direction) must be determined uniquely.
Equation (\ref{min1-11}) can be viewed as an {\it alternative form of the acoustic axes condition.}  Formally statement can be presented as follows:

\begin{prop}{\underline{\it Polarization vector  condition}}

    The reduced acoustic tensor $\bY=\bY({\bf n})$ has an acoustic axis in the direction ${\bf n}$ if and only if there is a vector ${\bf q}$ and the scalar $\sigma$ such that the equation 
    \begin{equation}\label{min1-6aa}
    \boxed{ \bY({\bf n})= \sigma(\bI-3{\bf q}\otimes{\bf q})}
    \end{equation}
   holds. Here, $\sigma$ represents the speed parameter, while the vector ${\bf q}$ represents the polarization vector corresponding to the single eigenvalue. 

\end{prop}


Observe that it is a well-posed system of 5 independent second-order algebraic equations for 5 independent variables: two components of the unit vector $\bf n$, two components of the unit vector $\bf q$, and the scalar parameter $\sigma$. 

\section{Relations between different conditions}
The above-discussed acoustic axes conditions are presented in rather different formats. The adjoint and minimal polynomial conditions are represented by matrices, the Norris condition by a third-order tensor, and the discriminant condition by a single scalar equation. In this section, we address the connections between different acoustic axis conditions.
\subsection{Equivalence between discriminant and minimal polynomial conditions}
Let us compare the matrix minimal polynomial condition with the scalar discriminant condition. 
\begin{prop}
    The discriminant condition 
    \begin{equation}\label{eq1}
        2({\rm det}\, {\pmb {\rm Y}})^2=\left(\frac 13\,{\rm tr}\, {\pmb {\rm Y}}^2\right)^3
    \end{equation}
    and the minimal polynomial condition 
\begin{equation}\label{eq2}
\bY^2+\sigma\bY-2\sigma^2\bI=0
\end{equation}
are equivalent to one another.
\end{prop}
\begin{proof}

(i) Let us assume that (\ref{eq2}) holds. Applying the trace operator to both sides of this equation, we get
\begin{equation}\label{eq3}
    {\rm tr}\,\bY^2=6\sigma^2.
\end{equation}
Multiplying  (\ref{eq2}) by the matrix $\bY$ and applying the trace operator once more we obtain
\begin{equation}\label{eq4}
{\rm tr}\,\bY^3+\sigma\,{\rm tr}\,\bY^2=0.
\end{equation}
Then 
\begin{equation}\label{eq5}
 \sigma=-\frac {{\rm tr}\,\bY^3}{{\rm tr}\,\bY^2}.  
\end{equation}
Eliminating the parameter $\sigma$ from the system (\ref{eq3},\ref{eq5}) we obtain 
\begin{equation}
    6\,\left({\rm tr}\,{\pmb {\rm Y}}^3\right)^2=\left({\rm tr}\, {\pmb {\rm Y}}^2\right)^3
\end{equation}
that is equivalent for the third order matrices to (\ref{eq1}). 

(ii) Let us assume that the scalar condition (\ref{eq1}) holds. Denote 
\begin{equation}
    \left(\frac 12{\rm det}\, {\pmb {\rm Y}}\right)^{1/3}=\tau
\end{equation}
Then we have three invariants of the matrix $\bY$
\begin{equation}
   {\rm tr}\, \bY=0,\qquad  {\rm tr}\, \bY^2= 6\tau^2  ,\qquad {\rm det}\,\bY=2\tau^3.
\end{equation}
With these invariants we construct the characteristic polynomial of the matrix $\bY$
\begin{equation}\label{}
f(\rho)=\rho^3-3\tau^2\rho-2\tau^3.
\end{equation}
This expression is factorized as 
\begin{equation}
    f(\rho)=(\rho+\tau)^2(\rho-2\tau).
\end{equation}
Consequently, the minimal polynomial takes the form
\begin{equation}
    m(\rho)=(\rho+\tau)(\rho-2\tau)
\end{equation}
The corresponding matrix equation takes the form
\begin{equation}
\left(\bY+\tau\bI\right)\left(\bY-2\tau\bI\right)=0.
\end{equation}
The latter equation  coincides with (\ref{eq2}) when we identify the parameters $\tau=-\sigma.$ 
\end{proof}
\subsection{Different conditions related}
Collecting the results derived  above, we depict the equivalency relations between the various kinds of acoustic axe conditions in Fig. 2. As we can see, the minimal polynomial construction is essential to this logical flow diagram since it establishes a connection between the various conditions.  

    \begin{figure}[H]\label{Fig2}
    \centering
\includegraphics[width=0.95\textwidth]{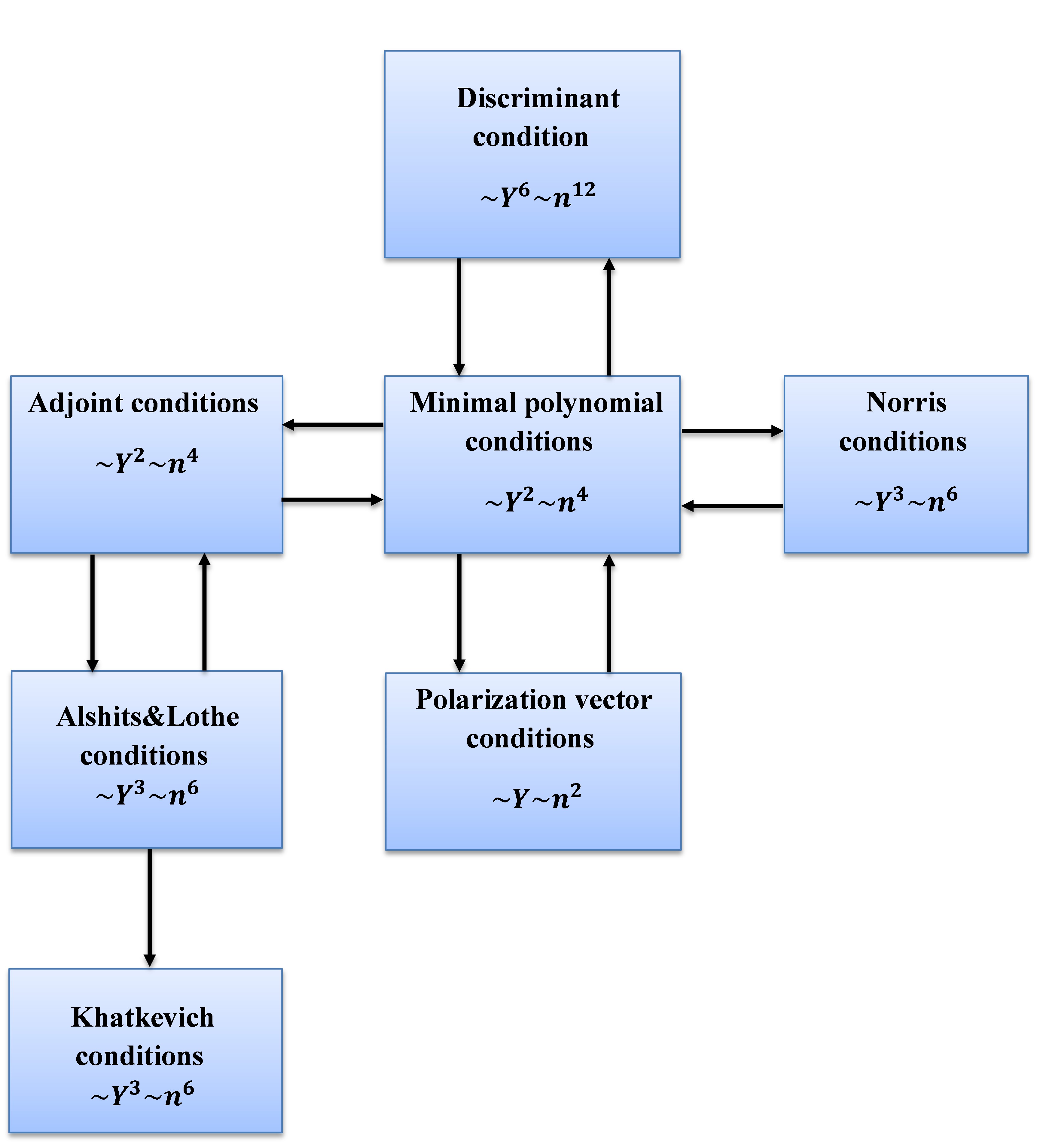}
\caption{Relationships between different kinds of acoustic conditions are shown in logical flow diagram. The polynomial orders are provided in terms of the matrix $\bY$and the vector $\bf n$.  
}
\end{figure}
\subsection{Computational aspects}
In this section, we discuss the computing algorithms that correlate to various presentations of the acoustic axis conditions. The acoustic axis problem is readily extended: in addition to finding the exact directions $\bf n$ of the acoustic axes, one further looks for the polarization vector ${\bf U}^{(3)}$ corresponding to the single eigenvalue and the speeds of three acoustic waves propagating in this direction.  The extra palarization vectors ${\bf U}^{(1)}$ and ${\bf U}^{(2)}$, which correspond to the double eigenvalue, are degenerate in the plane perpendicular to ${\bf U}^{(3)}$.

\subsubsection{12-th order conditions} 
The discriminant condition presented in Fedorov's form (\ref{ac-ax8}) or in the reduced form (\ref{ac-ax16}) is a sixth-order polynomial equation in the components of the Christoffel tensor $\bY$. Then it is a 12-th order polynomial equation in the components of the propagation vector ${\bf n}$. For the unit vector ${\bf n}$ with two independent components, this equation only has a final number of actual solutions in specific situations. 

 This 12th-order polynomial equation can be solved explicitly only in the case of the highest symmetry---isotropic media. In this case, the equation turns into an identity, see Sect.7. Then the propagation vector can be directed arbitrarily. The speeds of the waves are derived from the characteristic equation whereas the polarization vectors follow from the wave equation.  
\subsubsection{6-th order conditions} 
In Khatkevich's approach as well as its different versions, one is dealing with a system of third-order polynomial equations for the acoustic tensor components. In terms of the propagation vector components, we have here a pair of six-order homogeneous polynomial equations for two independent components of the propagation vector, say $n_1$ and $n_2$. 
As we already mentioned this system is completely the same when applied to the full matrix $\bG $ or to the reduced matrix $\bY$. In the traditional approach, this pair is derived from the matrix equation ${\rm adj}(\bG-\lambda\bI)=0$ by illumination of the parameter $\lambda$. As we demonstrated above, the same result follows from the reduced matrix equation ${\rm adj}(\bY-\lambda\bI)=0$.
The original pair of Khatkevich's equations is non-invariant and applicable only in the case when the out-of-diagonal components of the acoustic tensor are non-zero.  In the Alshits-Lothe approach, the pair of the Khatkevich conditions is extended to a system of seven conditions all of the sixth order. In a similar vein, the seven equations that satisfy the Norris condition are cubic in terms of the ${\bY}$-matrix components, or sixth order in terms of the vector ${\bf n}$. This system is capable of handling the exceptional cases and is explicitly invariant. 

\subsubsection{4-th order conditions} 
The 6-th order Khatkevich-type conditions can be obtained from the adjoint conditions by elimination of the scalar parameter $\sigma.$ Therefore, the adjoint matrix equation 
 \begin{equation}\label{min-vs2x}
   ({\rm adj}\,\bY)+\sigma\bY
   +\sigma^2\bI=0
\end{equation}
itself can serve as a base for a computation algorithm. 
 
The minimal polynomial conditions matrix system 
\begin{equation}\label{min-vs1x}
   \bY^2+\sigma\bY-2\sigma^2\bI=0
\end{equation}
appears to be even more straightforward when it comes to practical calculations.
  In both cases, one deals with a system of sixth equations for three independent variables---two components of the unit vector ${\bf n}$ and the speed parameter $\sigma$. In general, the diagonal equations follow from the out-of-diagonal ones, then the system appears to be well-posed. 
\subsubsection{2-nd order conditions} 
Using the solution to the minimal  polynomial equation, we can employ the polarization vector condition 
 \begin{equation}\label{min1-6x}
 \bY({\bf n})= \sigma(\bI-3{\bf q}\otimes{\bf q})
    \end{equation}
as the framework of a computation algorithm. This system consists of five independent equations for five variables: two components of the unit vector ${\bf n}$, two components of the unit vector ${\bf q}$, and the scalar $\sigma$. Thus, the system is well-posed. It is convenient to write it in the matrix form
\begin{equation}\label{min1-13}
\begin{pmatrix}
Y^{11} &Y^{12} &Y^{13}  \\
Y^{21} &Y^{22} &Y^{23}  \\
Y^{31} &Y^{32} &Y^{33}  
  \end{pmatrix}=-3\sigma \begin{pmatrix}
q_1^2-\third  &q_1q_2&q_1q_3 \\
q_1q_2 &q_2^2-\third&q_2q_3\\
q_1q_3 &q_2q_3  &q_3^2-\third 
  \end{pmatrix}\,.
\end{equation}
The terms in the matrices on the left and right sides can therefore be directly equated. The equations between the diagonal terms will be referred to as diagonal equations, and the equations between the out-of-diagonal terms as out-of-diagonal equations. 
Observe a useful relation between the diagonal and out-of-diagonal elements of the matrix $\bY$
\begin{equation}\label{min1-14}
    (Y^{11}-\sigma)(Y^{22}-\sigma)(Y^{33}-\sigma) =Y^{12}Y^{13}Y^{23}=-27\sigma^3(q_1q_2q_3)^2.
\end{equation}

The practical algorithm can be applied in the following sequence of steps: 
\begin{itemize}
    \item[(1)] Consider the polarization vector ${\bf q}$ that lies along the coordinate axes, for example, ${\bf q}=(1,0,0)$.  Then we have $Y^{12}=Y^{13}=Y^{23}=0$.  We may derive the vector ${\bf n}$ and the scalar $\sigma$ by combining these equations with the two independent diagonal equations. 
    \item[(2)] Take into consideration the polarization vectors ${\bf q}$, such as ${\bf q}=(q_1,q_2,0)$, that are located on the coordinate planes away from the coordinate axes. Then we have $Y^{13}=Y^{23}=0$. We obtain a system of four equations for four independent variables by combining these equations with the diagonal equations of (\ref{min1-13}). 
    \item[(3)] Consider the oblique polarization vectors ${\bf q}=(q_1,q_2,q_3)$ with the nonzero components.  The squared components $q_1^2, q_2^2$, and $q_3^2$ can be obtained by applying the relation (\ref{min1-14}). These are substituted into the two diagonal equations, yielding the scalar $\sigma$ and the components of the vector ${\bf n}$. 
\end{itemize}
We apply this algorithm to the case of the RTHC crystals in the next section. 
\section{Applications}
In this section, we examine some basic examples of acoustic waves in high-symmetric media. These results are widely recognized from the literature. Here, we aim to examine how the acoustic axes problem may be solved using the minimal polynomial approach with the polarization vector presentation. 
\subsection{Isotropic media}
To begin, let us examine the most basic instance of an isotropic elastic media. The elasticity tensor is therefore expressed in terms of the metric tensor $g^{ij}$ and two independent scalar parameters---the Lam\'e moduli $\lambda$ and $\mu$, see \cite{Landau}, \cite{Marsden}, 
\begin{equation}\label{iso1}
C^{ijkl}=\lambda\,g^{ij}g^{kl}+\mu\left(g^{ik}g^{jl}+g^{il}g^{jk}
  \right)\,.
\end{equation}
Then the Christoffel tensor depends linearly on two combinations  of media parameters,
\begin{equation}\label{iso-chris}
{\Gamma}^{ij}=\frac 1{\rho}\left(\mu g^{ij}+(\lambda+\mu)n^in^j\right)\,.
\end{equation}
To construct the depressed Christoffel tensor, we calculate  the trace of this matrix 
\begin{equation}
    \gamma=\frac 13 {\rm tr}\,\bG=\frac 13 {\Gamma}^{ij}g_{ij}=\frac 1{3\rho} (4\mu+\lambda).
\end{equation}
The depressed Christoffel tensor $\bY$ is presented by 
\begin{equation}\label{iso3}
Y^{ij}=\frac{\lambda+\mu}\rho\left(n^in^j-\frac 13 g^{ij}\right).
\end{equation}
Unlike the tensor $\bG$, 
the tensor $\bY$ is dependent on only one combination of media parameters.

Along the acoustic axis, the minimal polynomial condition (\ref{min1-11}) yields the following equation:
\begin{equation}\label{iso4}
\frac{\lambda+\mu}\rho\left(n^in^j-\frac 13 g^{ij}\right)=-3\sigma\left(q^iq^j-\frac 13 g^{ij}\right).
\end{equation}
One can derive equivalence between the unit vectors ${\bf n}$ and ${\bf q}$ by multiplying both sides of this equation by $q_j$. Following that, Eq.(\ref{iso4}) holds if and only if 
\begin{equation}\label{iso4x}
    {\bf n}={\bf q},\qquad \sigma=-\frac 1{3\rho}(\lambda+\mu).
\end{equation}
As a result, an arbitrary vector ${\bf n}$ direction serves as an acoustic axe. Moreover, the orientation of the polarization corresponding to the single eigenvalue is the same. This wave is longitudinal. Two additional shear waves match the double eigenvalue. 
The expressions for the wave speeds take the standard form: For the double eigenvalue---shear waves,
\begin{equation}
    v^2_S=\gamma+\sigma=\frac\mu\rho,
\end{equation}
 and for the single eigenvalue---longitudinal wave,
\begin{equation}
    v^2_L=\gamma-2\sigma=\frac 1{3\rho}(2\mu+\lambda).
\end{equation}
All axes are of the prolate type if the sum  $(\lambda+\mu)$ possess positive values, as indicated by Eq. (\ref{iso4x}). 
The positiveness of the squared velocities yields the bounds $\lambda+2\mu>0$ and $\mu>0$ that are less restrictive than the usual stability bounds: $\lambda+(2/3)\mu>0$ and $\mu>0$.
 
It is instructive to consider the discriminant condition for the acoustic axes. Let us calculate
\begin{equation}
    \bY^2=\frac {(\lambda+\mu)^2}3\left(n^in^j+\frac 13 g^{ij}\right), \qquad \bY^3=\frac {(\lambda+\mu)^3}3\left(n^in^j-\frac 19 g^{ij}\right).
\end{equation}
Then we have 
\begin{equation}
{\rm tr}\,\bY^2=\frac 23(\lambda+\mu)^2\,, \qquad {\rm tr}\,\bY^3=\frac 29(\lambda+\mu)^3.
\end{equation}
Consequently, the equation $6\left({\rm tr}\,{\pmb {\rm Y}}^3\right)^2=\left({\rm tr}\, {\pmb {\rm Y}}^2\right)^3$ holds identically for an arbitrary vector $\bf n$. We again arrive at a well-known result: every direction $\bf n$ serves as an acoustic axe in isotropic material.
\subsection{Cubic system}
Cubic crystals are described by three independent elasticity
constants $C_{11}, \, C_{44}$ and $C_{12}$. Correspondingly, the  Christoffel tenor (see  \cite{Musgrave}) 
\begin{equation}\label{cub-Cau-Christof}
  \Gamma^{il}=\frac 1\rho \begin{pmatrix}
\a n_1^2+\b & \d n_1n_2  & \d n_1n_3  \\
\g n_1n_2 &\a n_2^2+\b & \d n_2n_3  \\
    \d n_1n_3  & \d n_2n_3  & \a n_3^2+\b 
     \end{pmatrix}
\end{equation}
depends on three independent combinations 
\begin{equation}
   {\a}=C_{11}-C_{44},\quad {\b}= C_{44},\quad 
{\d}=C_{12}+C_{44}\,.
\end{equation}
The trace parameter is given by 
\begin{equation}
    \g=\frac{\a+3\b}{3\rho}.
\end{equation}
 Then the reduced Christoffel tensor turns out to be
\begin{equation}\label{cub-nonCau-Christof}
  Y^{il}=\frac \a \rho
  \begin{pmatrix}
n_1^2-\third &\xi n_1n_2&\xi n_1n_3 \\
\xi n_1n_2 &n_2^2-\third &\xi n_2n_3\\
\xi n_1n_3 &\xi n_2n_3  &n_3^2-\third
  \end{pmatrix} \,,
\end{equation}
where we introduce a constant factor 
\begin{equation}
    \xi=\frac{\d}{\a}=\frac{C_{12}+C_{44}}{C_{11}-C_{44}}.
\end{equation}
The tensor $\bY$ depends solely on two combinations of media characteristics. In contrast to the isotropic case provided by $\xi= 1$, the parameter $\xi\ne 1$ might be considered an anisotropy factor. 

Utilizing the minimal polynomial method, we require
\begin{equation}
\frac {\a}\rho 
  \begin{pmatrix}
n_1^2-\third  &\xi n_1n_2&\xi n_1n_3 \\
\xi n_1n_2 &n_2^2-\third&\xi n_2n_3\\
\xi n_1n_3 &\xi n_2n_3  &n_3^2-\third
  \end{pmatrix}= -3\sigma \begin{pmatrix}
q_1^2-\third &q_1q_2&q_1q_3 \\
q_1q_2 &q_2^2-\third&q_2q_3\\
q_1q_3 &q_2q_3  &q_3^2-\third 
  \end{pmatrix} \,.
\end{equation}
Two sides of this equation only coincide in two distinct cases:

(i) The out-of-diagonal terms vanish. Then we have three different acoustic axes:
\begin{equation}
{\bf n}=(1,0,0),\qquad {\bf n}=(0,1,0),\qquad {\bf n}=(0,0,1).
\end{equation}
The polarization vector and speed parameter fulfill, respectively,   
\begin{equation}
  {\bf q}={\bf n},\qquad   \sigma=-\frac \a{3\rho} =\frac 1{3\rho} (C_{44}-C_{11}).
\end{equation}
To have non-depressed waves in these directions, their squared speeds must be positive, implying that the inequality (\ref{bound}) holds. Then we have 
\begin{equation}
    \a+\b>0,\qquad \b>0,
\end{equation}
or, equivalently, 
\begin{equation}\label{bound1}
    C_{11}>0,\qquad C_{44}>0.
\end{equation}
Given that $C_{44}<C_{11}$ for the majority of cubic crystals, these axes are of the prolate type.

(ii) The diagonal terms vanish. Then we have four  different acoustic axes:
\begin{equation}\label{cub-ax}
{\bf n}=\frac{1}{\sqrt 3}(1,1,1),\qquad {\bf n}=\frac{1}{\sqrt 3}(-1,1,1),\qquad {\bf n}=\frac{1}{\sqrt 3}(1,-1,1),\qquad {\bf n}=\frac{1}{\sqrt 3}(1,1,-1). 
\end{equation}
The polarization vector and speed parameter fulfill, respectively,  
\begin{equation}
  {\bf q}={\bf n},\qquad  \sigma=-\frac {\xi\a} {3\rho}=-\frac {\d} {3\rho}=-\frac 1{3\rho} (C^{12}+C^{44}).
\end{equation}
From the inequality  (\ref{bound}) we have 
\begin{equation}
    -\frac 12(\a+3\b)<\d<\a+3\b,
\end{equation}
or, equivalently, 
\begin{equation}\label{bound2}
    -\frac 12 C_{11}-2C_{44}<C_{12}<C_{11}+C_{44}.
\end{equation}
Assuming $C_{12}+C_{44}>0$, the axes given in (\ref{cub-ax}) are likewise of the prolate type.

It is noteworthy to compare the bounds (\ref{bound1}) and (\ref{bound2}) with the ones given in the famous Born's stability criterion for cubic crystals:
\begin{equation}
    C_{11}-C_{12}>0,\qquad C_{11}+2C_{12}>0,\qquad C_{44}>0.
\end{equation}
Fig. 3 shows that bounds (\ref{bound1}) and (\ref{bound2}) are less restrictive than the necessary and sufficient Born's stability requirements.
    \begin{figure}[H]\label{Fig3}
    \centering
\includegraphics[width=0.35\textwidth]{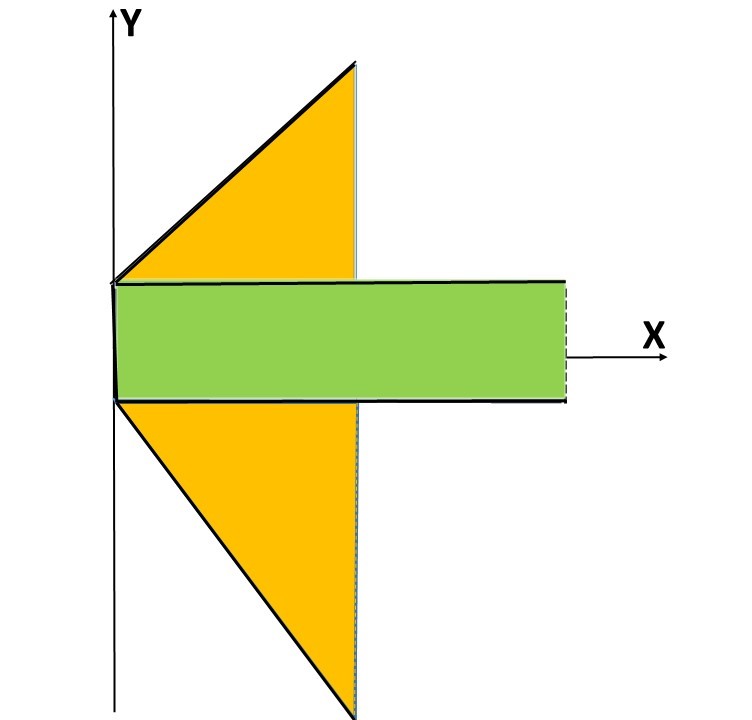}
\caption{We use the notations $X=C_{44}/C_{11}$  and $Y=C_{12}/C_{11}$. Then the bounds (\ref{bound1}) and (\ref{bound2}) are expressed as $-0.5-2X<Y<1+X$ while Born's bound is given by $-0.5<Y<1$. In both cases, $X>0$. 
}
\end{figure}

\subsection{RTHC crystals}
For all orthorhombic, tetragonal, hexagonal, or cubic crystals (‘RTHC crystals’), the diagonal and off-diagonal elements of the acoustic tensor $\Gamma^{ij}$ may be written, using appropriate Cartesian crystallographic axes, under a unique form; see Fedorov \& Fedorov \cite{Fed-Fed}, Musgrave \cite{Musgrave},  Boulanger \& Hayes \cite{Boulanger1}, \cite{Boulanger2}.  In this presentation, the diagonal elements are expressed as linear combinations of squares
\begin{equation}\label{RTHC1}
\Gamma^{11}=\sum^3_{j=1}a_{1j}n_j^2,\qquad \Gamma^{22}=\sum^3_{j=1}a_{2j}n_j^2,\qquad \Gamma^{33}=\sum^3_{j=1}a_{3j}n_j^2,
\end{equation}
while the out-of-diagonal elements are expressed as  monomials 
\begin{equation}\label{RTHC2}
\Gamma^{12}=r_{12}n_1n_2,\qquad \Gamma^{13}=r_{13}n_1n_3,\qquad 
    \Gamma^{23}=r_{23}n_2n_3.
\end{equation}
The six coefficients $a_{ij}=a_{ji}$ and the three coefficients $r_{ij}$ are linear combinations of the components of the elasticity tensor. Their explicit expressions for different symmetry classes can be find in \cite{Musgrave} and \cite{Boulanger2}. 
In particular, we observe from Eq.(\ref{cub-Cau-Christof}) that in the cubic case $a_{ii}=\a$, $a_{ij}=\b-\g$, and $r_{ij}=2\b+\g$. 

In order to examine the minimal polynomial approach, we express the components of the reduced acoustic tensor ${\bf Y}$ for the RTHC crystals. 
The trace of the acoustic tensor ${\pmb \Gamma}$ is given by 
\begin{equation}
    \gamma=\frac 13{\rm tr}\,\,{\pmb \Gamma}=\frac 13\sum^3_{j=1}(a_{1j}+a_{2j}+a_{3j})n_j^2.
\end{equation}
Correspondingly, the reduced acoustic tensor for the RTHC crystals has the components:
\begin{equation}\label{RTHC1a}
    Y^{11}=\sum^3_{j=1}b_{1j}n_j^2,\qquad Y^{22}=\sum^3_{j=1}b_{2j}n_j^2,\qquad 
    Y^{33}=\sum^3_{j=1}b_{3j}n_j^2,
\end{equation}
and
\begin{equation}\label{RTHC2x}
    Y^{12}=r_{12}n_1n_2,\qquad Y^{13}=r_{13}n_1n_3,\qquad 
    Y^{23}=r_{23}n_2n_3.
\end{equation}
Here the coefficients are given as 
\begin{equation}\label{RTHC1b}
    b_{1j}=\frac 13(2a_{1j}-a_{2j}-a_{3j}), \qquad b_{2j}=\frac 13(2a_{2j}-a_{1j}-a_{3j}), \qquad b_{3j}=\frac 13(2a_{3j}-a_{1j}-a_{2j}).
\end{equation}
Observe that the traceless property $Y^{11}+Y^{22}+Y^{33}=0$ is expressed in the following simple relation
\begin{equation}
    b_{1j}+b_{2j}+b_{3j}=0. 
\end{equation}
In the matrix form, the tensor ${\bf Y}$ is presented as 
\begin{equation}\label{RTHC3}
 {\bf Y}=\begin{pmatrix}
Y^{11} &r_{12}n_1n_2&r_{13} n_1n_3 \\
r_{12} n_1n_2 &Y^{22} &r_{23} n_2n_3\\
r_{13} n_1n_3 &r_{23} n_2n_3  &Y^{33}
  \end{pmatrix}.
\end{equation}
Using the polarization vector  representation, we have an equation 
\begin{equation}\label{RTHC4}
  \begin{pmatrix}
Y^{11} &r_{12}n_1n_2&r_{13} n_1n_3 \\
r_{12} n_1n_2 &Y^{22} &r_{23} n_2n_3\\
r_{13} n_1n_3 &r_{23} n_2n_3  &Y^{33}
  \end{pmatrix}=-3\sigma\begin{pmatrix}
q_1^2-\third &q_1q_2&q_1q_3 \\
q_1q_2 &q_2^2-\third&q_2q_3\\
q_1q_3 &q_2q_3  &q_3^2-\third 
  \end{pmatrix} \,.
\end{equation}
In general, this traceless symmetric system involves five independent equations for five variables---two components of the unit vector ${\bf n}$, two components of the unit vector ${\bf q}$, and the scalar $\sigma$. 

We start the solution Eq.(\ref{RTHC4}) with the special cases:

\vspace{0.2 cm}

(i) Let the acoustic axis coincides   with one of the coordinate axis. Then two components of the propagation vector $\bf n$ are equal to zero. Consider, for instance, ${\bf n}=(1,0,0)$. By  comparison the out-of-diagonal terms in Eq.(\ref{RTHC4}) we get three possibilities for the polarization vector ${\bf q}=(1,0,0),{\bf q}=(0,1,0)$ or ${\bf q}=(0,0,1)$. 
From the diagonal terms, we obtain for ${\bf n}={\bf q}=(1,0,0)$ 
\begin{equation}\label{RTHC4b}
    Y^{11}=\frac 13(2a_{11}-a_{21}-a_{31})=-2\sigma,
\end{equation}
and
\begin{equation}\label{RTHC4c}
    Y^{22}=\frac 13(2a_{21}-a_{11}-a_{31})=\sigma,\qquad Y^{33}=\frac 13(2a_{31}-a_{11}-a_{21})=\sigma.
\end{equation}
Consequently, to have an acoustic axe along the vector ${\bf n}=(1,0,0)$, the components of the elasticity tensor must satisfy 
\begin{equation}
    a_{21}=a_{31}\qquad {\rm and}\qquad \sigma=\frac 13(a_{21}-a_{11}).
\end{equation} 

In the alternative case ${\bf n}=(1,0,0)$ but ${\bf q}=(0,1,0)$, we have 
\begin{equation}\label{RTHC4bx}
    Y^{11}=\frac 13(2a_{11}-a_{21}-a_{31})=\sigma,
\end{equation}
and
\begin{equation}\label{RTHC4cx}
    Y^{22}=\frac 13(2a_{21}-a_{11}-a_{31})=-2\sigma,\qquad Y^{33}=\frac 13(2a_{31}-a_{11}-a_{21})=\sigma.
\end{equation}
Then we are left with a condition 
\begin{equation}
    a_{11}=a_{31}\qquad {\rm and}\qquad \sigma=\frac 13(a_{11}-a_{31}).
\end{equation} 
The similar types of restrictions on the elasticity tensor components appear for two other solutions ${\bf n}=(0,1,0)$ and ${\bf n}=(0,0,1)$. The results are presented in the following table

\begin{center}
\begin{tabular}{ |c|c|c|c|c|c|c|c| } 
 \hline
 ${\bf n}$ & {\bf q} & cond.&isotr.&cubic&hexag.&tetrag.&orthorh. \\ 
 \hline
$(1,0,0)$ & $(1,0,0)$ & $a_{21}=a_{31}$&\checkmark&\checkmark&-&-&- \\ 
  & $(0,1,0)$ & $a_{11}=a_{31}$&-&-&-&-&- \\
  & $(0,0,1)$ & $a_{11}=a_{21}$&-&-&-&-&- \\ 
$(0,1,0)$ & $(1,0,0)$ & $a_{22}=a_{32}$&-&-&-&-&-  \\ 
  & $(0,1,0)$ & $a_{12}=a_{32}$ &\checkmark&\checkmark&\checkmark&-&- \\
  & $(0,0,1)$ & $a_{12}=a_{22}$ &-&-&-&-&- \\ 
  $(0,0,1)$ & $(1,0,0)$ & $a_{23}=a_{33}$ &-&-&-&-&- \\ 
  & $(0,1,0)$ & $a_{13}=a_{33}$&-&-&-&-&-  \\
  & $(0,0,1)$ & $a_{13}=a_{23}$ &\checkmark&\checkmark&\checkmark&\checkmark&- \\ 
 \hline
\end{tabular}
\end{center}
Note that all acoustic axes are oriented in the same direction as the polarization vector ${\bf n}={\bf q}$, i.e., the waves are  longitudinal. Acoustic axes of this type do not exist for the  orthorhombic crystals. 
\vspace{0.2 cm}

(ii) Let us consider the acoustic axes on the coordinate planes out of the coordinate axes. Then only one component of the vector $\bf n$ is equal to zero. Consider, for instance, ${\bf n}=(n_1,n_2,0)$ with $n_1n_2\ne 0$. Comparing the out-of-diagonal terms we obtain Then ${\bf q}=(q_1,q_2,0)$ with $q_1q_2\ne 0$ and 
\begin{equation}\label{RTHC5}
    {r}_{12}n_1n_2=-3\sigma q_1q_2.
\end{equation}
From the diagonal terms we obtain the system
\begin{equation}\label{RTHC6}
    \begin{cases}
Y^{11}=-3\sigma q_1^2+\sigma \\
Y^{22}=-3\sigma q_2^2+\sigma\\
Y^{33}=\sigma.
    \end{cases}
\end{equation}
Eliminating the variables $q_1,q_2$ and $\sigma$, we are left with an equation that involves only the components of the vector ${\bf n}$
\begin{equation}\label{RTHC7}
    \left(Y^{11}-Y^{33}\right) \left(Y^{22}-Y^{33}\right)={r}^2_{12}n^2_1n^2_2
\end{equation}
Substituting here the expressions (\ref{RTHC1a}) we obtain the equation that is equivalent  to one given by Boulanger \& Hayes \cite{Boulanger2}, 
\begin{equation}
    \left((a_{11}-a_{31})n_1^2+(a_{12}-a_{32})n_2^2\right)\left((a_{21}-a_{31})n_1^2+(a_{22}-a_{32})n_2^2\right)=r_{12}^2n_1^2n_2^2.
\end{equation}
Thus we obtain the standard  quadratic equation $at^2+bt+c=0$ for the variable $t=(n_1/n_2)^2$ with the coefficients
\begin{eqnarray}\label{RTHC9}
    a&=&(a_{11}-a_{31})(a_{21}-a_{31})\nonumber\\
    c&=&(a_{12}-a_{32})(a_{22}-a_{32})\nonumber\\
    b&=&(a_{11}-a_{31})(a_{22}-a_{32})+(a_{12}-a_{32})(a_{21}-a_{31})-r^2_{12}\,.
\end{eqnarray}
The existence of the real positive solutions depends heavily on the values of the elasticity constants. In particular, for the cubic crystals, we have $a=c=0$. In this case, no axes lie on the coordinate planes out of the coordinate axes, as we stated in the previous section.

The speed parameter can be presented in the form
\begin{equation}\label{RTHC10}
\sigma=b_{31}n_1^2+b_{32}n_2^2.
\end{equation}

The maximal possible number of acoustic axes on the $xy$-plane (including the coordinate axes) is equal to 4. Thus, we have at most 12 axes lying on the coordinate planes.

\vspace{0.2 cm}

(iii) We consider now the acoustic axes lying out of the coordinate planes---oblique axes. 
Then all three components of the propagation vector $\bf n$ are nonzero. Comparing the out-of-diagonal terms in Eq.(\ref{RTHC4}) we deduce that all three components of the polarization  vector $\bf q$ are also nonzero. 
Then we have 
\begin{equation}\label{RTHC11}
    \begin{cases}
r_{12}n_1n_2=-3\sigma q_1q_2\\
r_{13}n_1n_3=-3\sigma q_1q_3\qquad\Longrightarrow\qquad 
r_{12}r_{13}r_{23}n^2_1n^2_2n^2_3=-27\sigma^3q^2_1q^2_2q^2_3\\
r_{23}n_2n_3=-3\sigma q_2q_3
    \end{cases}
\end{equation}
Observe that the prolate case (with $\sigma<0$) emerges for the elasticity modulus satisfying $r_{12}r_{13}r_{23}>0$, while the oblate case (with $\sigma>0$) is described by the inequality $r_{12}r_{13}r_{23}<0$. From Eqs.(\ref{RTHC11}) we obtain the proportionality relation between the components of the propagation vector ${\bf n}$ and of the single polarization vector ${\bf q}$
\begin{equation}\label{RTHC12b}
q_1^2=- \frac{\a_1}{3\sigma}n_1^2\,,\qquad 
q_2^2=-\frac{\a_2}{3\sigma}n_2^2\,,\qquad 
q_3^2=-\frac{\a_3}{3\sigma}n_3^2\,,
\end{equation}
where we denote  
\begin{equation}\label{RTHC12a}
\a_1=\frac{r_{12}r_{13}}{r_{23}},\qquad 
\a_2=\frac{r_{12}r_{23}}{ r_{13}},\qquad
\a_3=\frac{r_{13}r_{23}}{ r_{12}}\,.
\end{equation}

Compare now the diagonal terms on two sides of the system given in  Eq.(\ref{RTHC4}). For the $Y_{11}$-term we have 
\begin{equation}\label{RTHC12c}
    b_{11}{n_1^2} +b_{12}{n_2^2} +b_{13}{n_3^2}=\frac 23\a_1{n_1^2}-\frac 13\a_2{n_2^2}-\frac 13\a_3{n_3^2},
\end{equation}
or 
\begin{equation}\label{RTHC12d}
\left(b_{11}-\frac 23\a_1\right)n_1^2 
+\left(b_{12}+\frac 13\a_2\right)n_2^2 +\left(b_{13}+\frac 13\a_3\right)n_3^2
=0.
\end{equation}
Similarly, we have for the $Y_{22}$-term
\begin{equation}\label{RTHC12d1}
\left(b_{21}+\frac 13\a_1\right)n_1^2 
+\left(b_{22}-\frac 23\a_2\right)n_2^2 +\left(b_{23}+\frac 13\a_3\right)n_3^2
=0.
\end{equation}
We are able to interpret the left-hand sides of Eqs.(\ref{RTHC12d}--\ref{RTHC12d1}) as scalar products of the ``vector" $(n_1^2,n_2^2, n_3^2)$ with two vectors formed from the material parameters. from the material parameters. When the later two vectors are linearly independent,   the solution of the system is proportional to the vector product 
\begin{equation}\label{RTHC13}
    \begin{pmatrix}
n_1^2\\
n_2^2\\
n_3^2
\end{pmatrix}=C
\begin{pmatrix}
b_{11}-\frac 23\a_1\\
b_{12}+\frac 13\a_2\\
b_{13}+\frac 13\a_3
\end{pmatrix}\times
\begin{pmatrix}
b_{21}+\frac 13\a_1\\
b_{22}-\frac 23\a_2\\
b_{23}+\frac 13\a_3
\end{pmatrix},
\end{equation}
where $C$ is a normalized parameter. When the material vectors in the right-hand side are linearly independent, the solution is defined uniquely.  Since the right-hand-side  vectors depend on the material parameters only we derived here an explicit expression of the oblique acoustic axes.

The squared components of the vector $\bf q$ follow uniquely from Eq.(\ref{RTHC12b}). 
In order to derive the expression of the speed parameter $\sigma$ we write 
\begin{equation}\label{RTHC14}
1=q_1^2+q_2^2+q_3^2=- \frac{1}{3\sigma}\left(\a_1n_1^2+\a_2n_2^2+\a_3n_3^2\right).
\end{equation}
Then,
\begin{equation}\label{RTHC15}
\sigma=-\frac 13 \left(\a_1n_1^2+\a_2n_2^2+\a_3n_3^2\right).
\end{equation}
We obtain at most four oblique axes when we consider four possible sign combinations; the cubic case (\ref{cub-ax}) is one example of this. Then we come to a well-known result: the RTHC crystals can have at most 16 acoustic axes---four at every one of the coordinate planes and four out of the coordinate planes. 

Observe a special case when two material vectors on the left-hand-side of (\ref{RTHC13}) are proportional one to another or at least one of them equal zero. In particular, two matter vectors in (\ref{RTHC13}) are equal if the equations  
\begin{equation}\label{RTHC16}
    b_{11}=b_{21}+\a_1, \qquad b_{22}=b_{12}+\a_2, \qquad b_{13}=b_{23}
\end{equation}
hold. Straightforward substitution shows that these relations satisfied  for the hexagonal crystals. Indeed, in this case
the RTHC-parameters are expressed in terms of the elasticity modules as \begin{equation}
   b_{11}=b_{22}=\frac 13\left(2C_{11}-C_{66}-C_{44}\right),\qquad  b_{21}=b_{12}=\frac 13\left(2C_{66}-C_{11}-C_{44}\right)
\end{equation}
\begin{equation}
   \a_1=\a_2=C_{12}+C_{66},\qquad  b_{13}=b_{23}=\frac 13\left(C_{44}-C_{33}\right)
\end{equation}
Then Eqs.(\ref{RTHC16}) fulfil when the basic hexagonal relation $2C_{66}=C_{11}-C_{12}$ is applied. 

In this degenerate case, the components of the propagation vector ${\bf n}$ are restricted only by one equation (\ref{RTHC12d}) or (\ref{RTHC12d1}). We have 
\begin{equation}
    \left(C_{66}-C_{44}\right)(n_1^2+n_2^2)+\left(C_{44}-C_{33}+\frac{(C_{13}+C_{44})^2}{C_{12}+C_{66}}\right)n_3^2=0.
\end{equation}
If  the coefficients of this expression are of different signs, we come to a well-known result---the acoustic axes form a conic in the hexagonal crystals. We observe also an alternative situation: If all three coefficients in the quartic (\ref{RTHC12d}) are positive (negative) then oblique acoustic axes do not exist at all.

\section{Conclusion}
Conditions underlying the acoustic axis have been extensively studied for many years, and the acoustic axes have been precisely determined for all crystal symmetry classes higher than monoclinic. As noted by Norris \cite{Norris}, the "multiplicity of conditions" is the most important point to emphasize. Several systems of polynomial equations of different orders have been used in literature to describe these criteria. As far as we know, not enough research has been done on the relationships between these requirements, their necessity, and their invariant sense. 

In this work, we examine the relationships between the previously specified sets of  conditions.  Our approach relies on the reduced Christoffel tensor \cite{It-axes1}, which reduces the total number of elasticity parameters that show up within the acoustic axes conditions. 
We present a new set of conditions generated by using the minimal polynomial equation. Subsequently, the acoustic axes condition is expressed in a clear algebraic form: these are the directions in which the  third-order minimal  polynomial of $\bY$ reduces to the second order. We provide a solution to the minimal polynomial equation, which may be employed as an independent set of the acoustic axis conditions and presents the reduced Christoffel tensor by a unit vector and a scalar. The  scalar parameter represents the single eigenvalue, while the vector represents the direction of polarization. 

 We demonstrate how the several sets of requirements meet one another by employing the minimal polynomial conditions. Nonetheless, there are significant differences in the computational features of various sets of conditions. As a useful computational framework, we suggest the set of vector conditions. 
Applications of the technique for high symmetry RTHC-crystals are presented. Our primary goal here is to highlight the method's methodological and qualitative benefits. Hopefully, the techniques discussed here will prove helpful in explicitly computing the acoustic axes in the lower symmetric crystals.

\end{document}